\documentclass[
 amsmath,amssymb,
 aps, reprint,
prl
]{revtex4-2}

\usepackage{etoolbox}

\makeatletter
\patchcmd{\@bibdataout@aps}
  {author="08"}
  {author="48"}
  {}
  {\PackageWarning{main}{REVTeX jnrlst patch failed}}
\makeatother

\usepackage{graphicx}
\usepackage{dcolumn}

\usepackage{physics}
\usepackage{amsfonts, amsmath}
\usepackage{dsfont}
\usepackage{amsthm}
\usepackage{bm}

\usepackage{bbold}
\usepackage{color}
\usepackage[normalem]{ulem}
\usepackage{tabularx}

\usepackage{comment}

\usepackage{tikz}

\newtheorem{theorem}{Theorem}
\newtheorem{corollary}{Corollary}

\newtheorem{definition}{Definition}

\usepackage{hyperref}

\begin{document}
\definecolor{navy}{RGB}{46,72,102}
\definecolor{pink}{RGB}{219,48,122}
\definecolor{grey}{RGB}{184,184,184}
\definecolor{yellow}{RGB}{255,192,0}
\definecolor{grey1}{RGB}{217,217,217}
\definecolor{grey2}{RGB}{166,166,166}
\definecolor{grey3}{RGB}{89,89,89}
\definecolor{red}{RGB}{255,0,0}

\title{Bound Entanglement Is Insufficient for an Exponential Quantum Learning Advantage}

\author{Hyeongu Kang}

\author{Sangwoo Jeon}

\author{Changhun Oh}
\email{changhun0218@gmail.com}
\affiliation{Department of Physics, Korea Advanced Institute of Science and Technology, Daejeon 34141, Korea}

\begin{abstract}
While entanglement is known to enable exponential improvements in the sample complexity of quantum learning, it remains unclear which properties of entangled resources are responsible for such improvements.
We address this question through the reduction criterion, a condition obeyed by all bound-entangled states.
In $n$-qubit Pauli-channel learning, we show that restricting either the input states or the measurement effects to satisfy this criterion rules out an exponential advantage for incoherent adaptive protocols.
An exponential lower bound persists for the one-sided coherent adaptive protocols considered here, even when the unrestricted side retains quantum correlations across channel uses.
Using conditional min-entropy, we further quantify how the sample-complexity lower bounds weaken as larger violations of the reduction criterion are allowed.
Finally, we show that the same obstruction appears in conjugate-state learning: restricted joint measurements cannot reproduce the logarithmic-sample advantage of unrestricted joint measurements on $\rho\otimes\rho^*$. 
These results identify violation of the reduction criterion as a necessary condition for an exponential advantage in the learning tasks considered here.
\end{abstract}

\maketitle 

\textit{Introduction.{\textemdash}} 
Quantum entanglement is a central resource in quantum information science, enabling information-processing capabilities beyond what is possible with classical resources alone~\cite{Horodecki_2009}. 
It underlies quantum speedups in computation~\cite{nielsen2010quantum}, enhanced precision in metrology~\cite{giovannetti2006quantum, giovannetti2011advances, polino2020photonic, Degen_2017}, and intrinsically quantum forms of communication~\cite{gisin2007quantum, kimble2008quantum}.
These examples establish entanglement not merely as a signature of nonclassicality, but as an operational resource whose usefulness is revealed through concrete tasks.

Quantum learning is one such setting, concerned with extracting useful information about unknown quantum systems from measurement data~\cite{aaronson2018shadow, huang2020predicting, huang2022quantum,  huang2021information,    caro2024learning, chen2023complexity}. 
It is essential for the characterization, verification, and control of quantum systems~\cite{gebhart2023learning}, yet it becomes increasingly demanding at large scales, where many natural learning tasks require exponentially many samples~\cite{Aaronson_2007, o2016efficient, haah2017sample, cramer2010efficient}. 
Entanglement can be a key resource for overcoming this bottleneck, enabling protocols with entangled inputs~\cite{PhysRevLett.133.230604, Liu_2025}, quantum memory~\cite{chen2022exponential}, or joint measurements~\cite{bubeck2020entanglement, chen2023does, aharonov2022quantum} that achieve exponential reductions in sample complexity.
Among these problems, Pauli channel learning provides a central benchmark for this phenomenon~\cite{Chen_2022, PhysRevLett.132.180805, PRXQuantum.6.020323, chen2023learnability}.
Pauli channels form a basic noise model in quantum information~\cite{nielsen2010quantum, flammia2020efficient, Flammia_2021}, and learning their eigenvalues exhibits an exponential separation between entanglement-assisted~\cite{Chen_2022} and entanglement-free schemes~\cite{PhysRevLett.132.180805, PRXQuantum.6.020323}.

This separation raises a sharper question about the underlying resource.
Recent work shows that its exponential advantage over entanglement-free strategies can persist even when the amount of entanglement in each input state becomes asymptotically small~\cite{Kim_2026}. 
Thus, the relevant resource cannot be characterized solely by the amount of entanglement. 
Rather, the advantage appears to depend on qualitative features of the entangled resource, motivating the central question of this work:
\begin{center}
\emph{What properties of entangled resources are necessary for an exponential quantum learning advantage?}
\end{center}

To address this question, we consider bound entanglement~\cite{Horodecki_1998, PRXQuantum.3.010101, amselem2013experimental, Hiesmayr_2025}, a distinguished class of genuinely entangled but nondistillable states.
Their nondistillability does not make them operationally equivalent to separable states, since bound-entangled states can outperform separable states in several information-processing tasks~\cite{Horodecki_1999, PhysRevLett.96.150501, PhysRevLett.134.120203, PhysRevLett.112.110502, singh2026ancillaassistedquantumprocess} such as channel discrimination~\cite{PhysRevLett.102.250501} and metrology~\cite{PhysRevResearch.3.023101}. 
It is therefore not evident whether they can retain the exponential learning advantage achievable with unrestricted entangled resources.

\begin{figure}[t]
    \centering
    \includegraphics[width=1.0\linewidth]{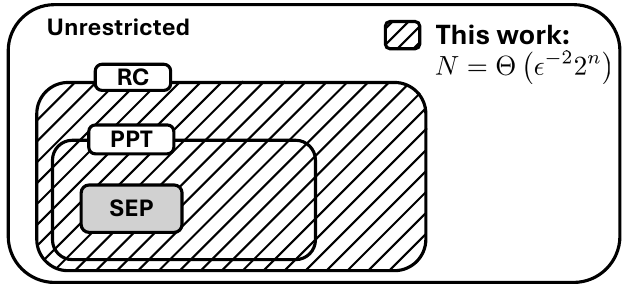}
\caption{ Resource hierarchy for incoherent adaptive Pauli-channel learning. For $0<\epsilon \leq \frac{1}{3},$ the tight scaling $N=\Theta \left( \epsilon^{-2} 2^{n} \right),$ previously known for SEP~\cite{PhysRevLett.132.180805, PRXQuantum.6.020323}, extends to the hatched one-sided reduction criterion (RC) class; the other side is unrestricted, and PPT is shown for reference.}
\label{fig_Van Diagram}
\end{figure}

Our results establish a sharp necessary condition for an exponential quantum learning advantage.
We prove that if either the system--ancilla input states or the measurement effects satisfy the reduction criterion~\cite{PhysRevA.59.4206, PhysRevA.60.898, Chitambar_2024}---a condition obeyed by all bound-entangled states---then Pauli-channel learning cannot retain an exponential improvement over entanglement-free strategies.
This obstruction holds for incoherent adaptive protocols, and an exponential lower bound continues to hold in the one-sided coherent adaptive models considered here.
Thus, an exponential improvement requires both the input and measurement resources to go beyond the reduction-criterion regime, as summarized in Fig.~\ref{fig_Van Diagram}.

We further make this boundary quantitative using conditional min-entropy~\cite{Tomamichel_2011}.
Rather than only distinguishing resources that satisfy or violate the reduction criterion, our bounds track how the sample complexity changes with the extent of the violation.
This yields a continuous family of lower bounds for restrictions imposed on either the input or measurement side.

We further examine conjugate-state learning, where the known advantage arises from joint measurements on $\rho\otimes\rho^*$~\cite{King_2024}. 
Unrestricted joint measurements on $\rho\otimes\rho^*$ can learn all displacement amplitudes of an unknown $d$-dimensional state using $O(\epsilon^{-4}\log d)$ samples.
For copywise adaptive protocols under the same reduction-criterion restriction on the normalized POVM elements, we prove a lower bound of $\Omega(\epsilon^{-2}\sqrt d)$.
This shows that the reduction-criterion restriction also limits a state-learning advantage generated entirely by measurement-side entanglement.

\medskip
\textit{Pauli channel learning.{\textemdash}} 
We first formalize the Pauli channel learning task.
Let $\mathcal{P}_n=\{I,X,Y,Z\}^{\otimes n}$ denote the set of
$n$-qubit Pauli operators.
An $n$-qubit Pauli channel is a CPTP map of the form
\begin{align}
    \Lambda(\rho)
    =
    \sum_{P\in\mathcal{P}_n}p_P P\rho P,
\end{align}
where $\{p_P\}$ is a probability distribution.
The coefficients $p_P$ are the Pauli error probabilities.
Each Pauli operator is an eigenoperator, $\Lambda(P)=\lambda_P P$.

Given $N$ uses of an unknown $n$-qubit Pauli channel, the learner aims to estimate all of its Pauli eigenvalues simultaneously:
\begin{definition}[Pauli channel learning]
Given $N$ uses of an unknown $n$-qubit Pauli channel $\Lambda$, 
the learner must output estimates $\{\hat{\lambda}_P\}_{P\in\mathcal{P}_n}$ such that, for every $\Lambda$, 
\begin{align}
    \max_{P\in\mathcal{P}_{n}} \left| \hat{\lambda}_P -\lambda_P \right| & \leq \epsilon
\end{align}
with probability at least $2/3$. 
The sample complexity is the minimum number $N$ of channel uses required to satisfy this guarantee. 
\end{definition}
\noindent

\begin{figure}[t]
    \centering
    \includegraphics[width=1.0\linewidth]{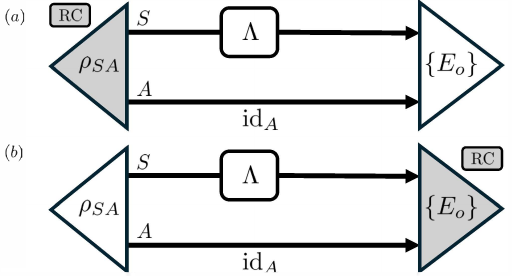}
    \caption{
Ancilla-assisted Pauli-channel learning with an incoherent adaptive protocol.
The unknown channel $\Lambda$ acts on the system register $S$, while the ancilla $A$ is retained noiselessly through $\mathrm{id}_A$.
The resulting joint state is then measured using the POVM $\{E_o\}$.
We consider one-sided restrictions in which either
(a) the input state $\rho_{SA}$ or
(b) each normalized POVM element satisfies the reduction criterion.
}
    \label{fig_PauliChannel_Learning}
\end{figure}

\medskip
\textit{Reduction-criterion restriction.{\textemdash}}
We impose a one-sided resource restriction on either the system--ancilla input states or the measurement effects.
For an input state $\rho_{SA}$, the reduction criterion requires
\begin{align}
    \rho_{SA}\le I_S\otimes \rho_A,
    \qquad
    \rho_A=\Tr_S\rho_{SA}.
\end{align}
Every bound-entangled state satisfies this criterion~\cite{PhysRevA.59.4206, PhysRevA.60.898}.
For a POVM $\{E_o\}$ on $SA$, we impose the analogous condition on each normalized effect:
\begin{align}
    \frac{E_o}{\Tr E_o}
    \le
    I_S\otimes \sigma_A^o,
    \qquad
    \sigma_A^o
    =
    \frac{\Tr_S E_o}{\Tr E_o},
\end{align}
for every outcome with $\Tr E_o>0$.
Thus, the measurement-side restriction also contains all bound-entangled normalized effects.

\begin{figure*}[t]
    \centering
    \includegraphics[width=0.97\textwidth]{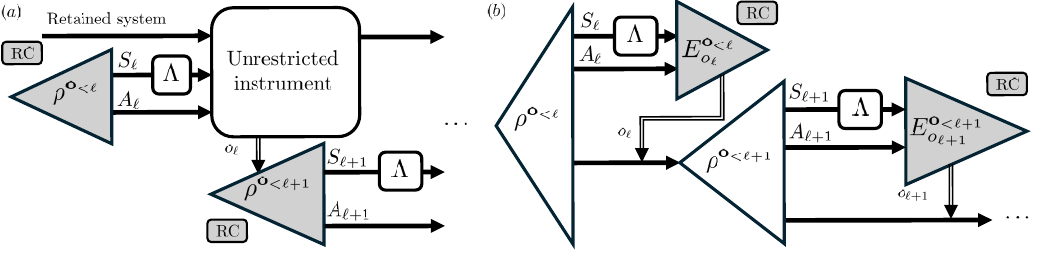}
\caption{
Coherent adaptive protocols under one-sided reduction-criterion (RC) restrictions.
(a) Input-restricted model: conditioned on the previous measurement outcomes $\mathbf{o}_{<\ell}$, the fresh input $\rho^{\mathbf{o}_{<\ell}}$ satisfies RC and is uncorrelated with the retained quantum memory; after the channel use, the output and memory may undergo arbitrary joint coherent processing.
(b) Measurement-restricted model: inputs may be entangled across channel uses, but the round-$\ell$ POVM acts only on $S_\ell A_\ell$, with each normalized POVM element satisfying RC.
}
    \label{fig_quantum_adaptivity}
\end{figure*}

We first consider protocols in which the input state and measurement in each channel use may depend on all previous measurement outcomes incoherently; no quantum correlation can be shared across the different channel uses.
Our first result shows that imposing the reduction-criterion restriction on either the input side or the measurement side removes the exponential learning advantage. 
The unrestricted side may still use arbitrary system--ancilla entanglement within each channel use as shown in Fig.~\ref{fig_PauliChannel_Learning}.
This freedom does not improve the optimal sample complexity over entanglement-free protocols.

\begin{theorem} \label{Thm_PauliChannel_ClassicalAdaptivity}
Assume $0<\epsilon\leq 1/3$.
Consider an incoherent adaptive protocol.
Suppose that one of the following holds throughout the protocol: either every system--ancilla input state satisfies the reduction-criterion restriction, or every measurement POVM satisfies the reduction-criterion restriction.
Then the optimal sample complexity for Pauli-channel learning is
\begin{align}
    N = \Theta\!\left(\epsilon^{-2}2^n\right).
\end{align} 
\end{theorem}
Since every bound-entangled state satisfies the reduction criterion, Theorem~\ref{Thm_PauliChannel_ClassicalAdaptivity} applies in particular to all bound-entangled inputs and normalized measurement effects.
This matches the known optimal sample complexity of entanglement-free schemes~\cite{PhysRevLett.132.180805, PRXQuantum.6.020323}.

\begin{proof}[Proof sketch]
The full proof is given in the Supplemental Material~\cite{supple}. 
The upper bound is achieved by the optimal entanglement-free protocol~\cite{PhysRevLett.132.180805, PRXQuantum.6.020323}.
For the lower bound, we use the standard randomly spiked Pauli-channel testing problem~\cite{PhysRevLett.132.180805, PRXQuantum.6.020323, Kim_2026}, in which the learner must detect a perturbation hidden in a uniformly random Pauli direction.
The key new ingredient is to show that the response bound underlying the entanglement-free lower bound remains valid under our one-sided reduction-criterion restriction.

The inequality $\rho_{SA}\leq I_S\otimes \rho_A$ bounds the average squared response of an input state to a random Pauli perturbation by $O(2^{-n})$, even when the state is entangled.
The analogous condition on normalized POVM elements gives the same bound when the restriction is imposed on the measurement side.
Because this estimate holds after conditioning on any previous measurement outcomes, it applies round by round to incoherent adaptive protocols.
Consequently, the total variation distance between the two measurement-outcome distributions is $O(N\epsilon^2 2^{-n})$, yielding $N=\Omega(\epsilon^{-2}2^n)$.
\end{proof}

Theorem~\ref{Thm_PauliChannel_ClassicalAdaptivity} shows that restricting either side of an incoherent adaptive protocol to the reduction-criterion regime eliminates the exponential advantage.
The other side may still use arbitrary system--ancilla entanglement within each channel use, yet the optimal sample complexity remains at the entanglement-free scaling.
Thus, an exponential improvement requires both the input and measurement resources to violate the reduction criterion.

We next consider two one-sided coherent adaptive models, which include our previous model.
In the input-restricted model, conditioned on prior outcomes,
each round uses a fresh input satisfying the reduction criterion
and uncorrelated with the retained quantum memory, while the
output and memory may subsequently be processed coherently.
In the measurement-restricted model, inputs may be entangled
across channel uses, but in each round the POVM acts only on
the current output and its associated ancilla, with every
normalized POVM element satisfying the reduction criterion.

\begin{theorem}
\label{Thm_PauliChannel_QuantumAdaptivity}
Assume $0<\epsilon\leq 1/3$.
Consider either of the two one-sided coherent adaptive models illustrated in Fig.~\ref{fig_quantum_adaptivity}.
In the input-restricted model, each round uses a fresh reduction-criterion-restricted input that is uncorrelated with the retained quantum system conditioned on the previous measurement outcomes, while arbitrary coherent processing of accumulated outputs is allowed.
In the measurement-restricted model, the inputs may be entangled across channel uses, while in each round a reduction-criterion-restricted POVM acts only on the current channel output and its associated ancilla.
In either model, Pauli-channel learning requires
\begin{align}
    N=\Omega\!\left(\epsilon^{-1}2^{n/2}\right)
\end{align}
channel uses.
\end{theorem}

Theorem~\ref{Thm_PauliChannel_QuantumAdaptivity} shows that coherent correlations on the unrestricted side cannot compensate for a reduction criterion restriction on the other side.
Although the resulting quantitative bound differs from the incoherent-adaptive scaling, an exponential obstruction persists in both settings.

Unlike the incoherent adaptive proof, this argument cannot be reduced to a round-by-round second-moment estimate.
We instead combine a hybrid argument with the leftover-hash lemma~\cite{Tomamichel_2011}.

\medskip
\textit{Quantitative relaxation.{\textemdash}}
We next quantify how the lower bounds change when the restricted resources are allowed to violate the reduction criterion.
For a density operator $\rho_{SA}$, the conditional min-entropy is defined as follows.

\begin{definition}[Conditional min-entropy~\cite{Tomamichel_2011}]
For a density matrix $\rho_{SA}$, the conditional min-entropy of the system $S$ given the ancilla $A$ is
\begin{align}
    H_{\min}(S\mid A)_\rho
    :=
    \max_{\sigma_A}
    \sup\left\{
        \kappa\in\mathbb{R}:
        \rho_{SA}\le 2^{-\kappa} I_S\otimes\sigma_A
    \right\},
\end{align}
where the maximization is over density operators $\sigma_A$.
\end{definition}
The reduction criterion implies $H_{\min}(S\mid A)_\rho\ge 0$. More generally, the condition $H_{\min}(S\mid A)_\rho\ge -\eta$ allows controlled violations of the reduction criterion, with larger $\eta$ permitting a broader class of states.

For measurements, we impose the analogous relaxation on normalized POVM elements. 
We say that a POVM $\{E_o\}$ on $SA$ satisfies 
$H_{\min}(S\mid A)_{\{E_o\}}\ge -\eta$ if, for every outcome $o$ with $\Tr E_o>0$, there exists a density operator $\sigma_A^o$ such that
\begin{align}
    \frac{E_o}{\Tr E_o}
    \le
    2^\eta I_S\otimes\sigma_A^o,
    \qquad
    \sum_o \Tr(E_o)\sigma_A^o
    =
    d_S I_A ,
\end{align}
where $d_S=2^n$. 
The first condition applies the same conditional-min-entropy relaxation to each normalized measurement effect, while the second preserves the completeness structure of the POVM. 
When no confusion arises, we use the notation $H_{\min}(S\mid A)$ for both input states and POVMs.

With this relaxation, the same proof yields the following quantitative lower bound for both incoherent and coherent adaptive protocols.

\begin{corollary}
\label{Cor_PauliChannel_Adaptivity}
Assume $0<\epsilon\leq 1/3$ and $\eta \geq 0$.
Consider either the incoherent adaptive model or one of the two one-sided coherent adaptive models introduced above.
Suppose that the restricted input states or normalized measurement effects satisfy
$H_{\min}(S\mid A)\ge -\eta$.
Then Pauli channel learning requires
\begin{align}
N
=
\begin{cases}
\Omega\!\left(\epsilon^{-2}2^{n-\eta}\right),
& \text{incoherent},\\[2mm]
\Omega\!\left(\epsilon^{-1}2^{(n-\eta)/2}\right),
& \text{coherent}.
\end{cases}
\end{align}
\end{corollary}

Corollary~\ref{Cor_PauliChannel_Adaptivity} shows that the obstruction is quantitative rather than discrete: the lower bounds weaken continuously as stronger violations of the reduction criterion are allowed.

\medskip
\textit{Conjugate-state learning.{\textemdash}}
We next consider conjugate-state learning, where the known advantage arises entirely from joint measurements on $\rho\otimes\rho^*$.
Unlike Pauli-channel learning, this setting isolates the role of entanglement on the measurement side. 

Let $X\ket{j}=\ket{j+1\!\!\!\pmod d}$.
Define $Z\ket{j}=e^{2\pi i j/d}\ket{j}$, $D_{q,p}:=e^{i\pi qp/d}X^qZ^p$ for $q,p\in\mathbb{Z}_d$, and let $y_{q,p}:=\Tr(D_{q,p}\rho)$ for an unknown $d$-dimensional state $\rho$~\cite{King_2024}.
Given copies of $\rho\otimes\rho^*$, the task is to estimate all $\{y_{q,p}\}$ to additive error $\epsilon$ with high success probability.

Conjugate-state learning exhibits an exponential sample-complexity advantage from joint entangled measurements. 
With suitable joint measurements on $\rho\otimes\rho^*$, all displacement amplitudes can be learned using $N=O(\epsilon^{-4}\log d)$ copies, whereas without joint measurements, even with access to both $\rho$ and $\rho^*$, one needs $N=\Omega(d/\epsilon^2)$ copies~\cite{King_2024}. 
Thus, the advantage is generated by measurement-side entanglement.

We consider incoherent adaptive protocols in which each copy of $\rho\otimes\rho^*$ is measured separately, while the joint POVM applied in each round may depend on all previous measurement outcomes.
We impose the same reduction-criterion restriction as before on the normalized POVM elements used in every round.
The following result shows that measurements within this class cannot retain the logarithmic dependence on $d$ achievable with unrestricted joint measurements.
The present proof controls each restricted measurement separately and does not address collective measurements across distinct copies.

\begin{figure}[t]
    \centering
    \includegraphics[width=0.47\textwidth]{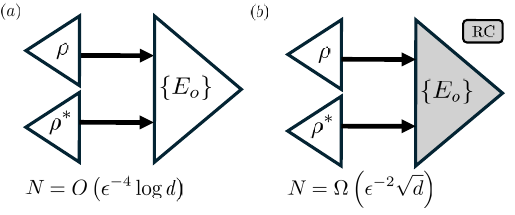}
\caption{
Conjugate-state learning under unrestricted and RC-restricted joint measurements.
(a) With unrestricted measurements~\cite{King_2024}, we can learn all displacement amplitudes from copies of $\rho\otimes\rho^*$ with
$N=O(\epsilon^{-4}\log d)$.
(b) For the copywise adaptive protocols considered here, the RC restriction yields the lower bound $N=\Omega(\epsilon^{-2}\sqrt d)$ for protocols that succeed with probability at least $2/3$.
}
    \label{fig_state_learning}
\end{figure}

\begin{theorem} \label{Thm_ConjugateStateLearning}
Assume $0<\epsilon<1/6$.
For the incoherent adaptive protocols described above, suppose that the normalized elements of every joint POVM satisfy the reduction-criterion restriction.
Then learning all displacement amplitudes of an unknown $d$-dimensional state with success probability at least $2/3$ requires
\begin{align}
    N=\Omega\left(\epsilon^{-2}\sqrt d\right)
\end{align}
copies.
\end{theorem}
Within the copywise adaptive protocols considered here, the reduction criterion is obeyed by all bound-entangled normalized effects, so Theorem~\ref{Thm_ConjugateStateLearning} rules out the logarithmic advantage for bound-entangled joint measurements.

As in the Pauli-channel setting, the measurement restriction can be relaxed quantitatively using conditional min-entropy, yielding a continuous family of lower bounds.

\begin{corollary} \label{Cor_ConjugateStateLearning}
Assume $0<\epsilon<1/6$ and $\eta \geq0$.
Suppose that every round of a copywise adaptive protocol described above uses a POVM satisfying $H_{\min}(S\mid A)\geq-\eta$ in the sense defined above.
Then any such protocol that succeeds with probability at least $2/3$ requires
\begin{align}
    N
    =
    \Omega\!\left(\epsilon^{-2}\sqrt{d\,2^{-\eta}}\right)
\end{align}
copies.
\end{corollary}

At $\eta=0$, Corollary~\ref{Cor_ConjugateStateLearning} recovers the lower bound of Theorem~\ref{Thm_ConjugateStateLearning}.
More generally, the lower bound decreases as stronger violations of the reduction criterion are allowed. 
Within the copywise adaptive model, bound-entangled joint measurements cannot retain the logarithmic-sample advantage of unrestricted joint measurements.

\medskip
\textit{Conclusion.{\textemdash}}
We have identified violation of the reduction criterion as a necessary resource property for the exponential learning advantages considered here.
In Pauli-channel learning, imposing the criterion on either the input or measurement side rules out an exponential improvement, even though the unrestricted side may retain arbitrary system--ancilla entanglement.
An exponential obstruction persists under the one-sided coherent adaptive models.
In conjugate-state learning, the same restriction on joint measurements prevents the logarithmic-sample advantage of unrestricted entangled measurements.
The obstruction is therefore not specific to a single learning task or to a particular placement of entanglement.

These results sharpen the distinction between the amount and the operational structure of entanglement.
Together with the observation that asymptotically small entanglement can suffice for exponential advantage, they show that entanglement alone does not determine the usefulness of a learning resource.
Our conditional-min-entropy bounds further demonstrate that this distinction is quantitative: the obstruction weakens continuously as larger violations of the reduction criterion are permitted.

Whether the square-root gap between the coherent and incoherent lower bounds is fundamental remains open.
We suspect that it may instead arise from repeated applications of data-processing inequalities in the coherent analysis.

Recent axiomatic work connecting the reduction criterion to nonnegative quantum conditional entropies~\cite{Gour_2024} places this boundary in a broader information-theoretic context.
A natural next question is to identify resource properties that are not only necessary but also sufficient for exponential quantum learning advantage, and to determine how broadly similar entropy-based limitations extend to other learning tasks.




\begin{acknowledgments} 
This work was supported by the National Research Foundation of Korea Grants (No. RS-2024-00431768 and No. RS-2025-00515456) funded by the Korean government (Ministry of Science and ICT (MSIT)) and the Institute of Information \& Communications Technology Planning \& Evaluation (IITP) Grants funded by the Korean government (MSIT) (No. RS-2024-00437284, No. IITP-2025-RS-2025-02283189 and No. IITP-2025-RS-2025-02263264) by Global Partnership Program of Leading Universities in Quantum Science and Technology (RS-2025-08542968) through the National Research Foundation of Korea~(NRF) funded by the Korean government (Ministry of Science and ICT (MSIT)).
\end{acknowledgments}

\nocite{lecam1973convergence, asadian2016heisenberg, Terras1999}

\bibliography{reference}

\end{document}


\definecolor{navy}{RGB}{46,72,102}
\definecolor{pink}{RGB}{219,48,122}
\definecolor{grey}{RGB}{184,184,184}
\definecolor{yellow}{RGB}{255,192,0}
\definecolor{grey1}{RGB}{217,217,217}
\definecolor{grey2}{RGB}{166,166,166}
\definecolor{grey3}{RGB}{89,89,89}
\definecolor{red}{RGB}{255,0,0}

\title{Supplemental Material for ``Bound Entanglement Is Insufficient for an Exponential Quantum Learning Advantage''}

\author{Hyeongu Kang}
\author{Sangwoo Jeon}

\author{Changhun Oh}
\email{changhun0218@gmail.com}
\affiliation{Department of Physics, Korea Advanced Institute of Science and Technology, Daejeon 34141, Korea}

\maketitle
\tableofcontents

\section{Bounds for incoherent adaptive Pauli-channel learning} \label{SM_Section_1}

In this section, we prove Theorem~1 and the incoherent-adaptive part of Corollary~1 of the main text.
We first recall the precise statements proved in this section.
\medskip

\noindent\textbf{Theorem 1 (main text).}
Assume $0<\epsilon\leq1/3$.
Consider an incoherent adaptive protocol.
If, in every channel use, either the system--ancilla input state or the measurement POVM satisfies the reduction-criterion restriction, the optimal sample complexity for Pauli-channel learning is
\begin{align}
N
=
\Theta\left(
\epsilon^{-2}2^n
\right).
\end{align}

\medskip
\noindent\textbf{Corollary 1 (incoherent-adaptive part; main text).}
Assume $0<\epsilon\leq1/3$.
Suppose that, in every channel use of an incoherent adaptive protocol, the restricted input states or normalized measurement effects satisfy
$H_{\min}(S\mid A)\geq-\eta$ for some fixed $\eta \geq0$.
Then Pauli-channel learning requires
\begin{align}
N
=
\Omega\left(
\epsilon^{-2}2^{n-\eta}
\right).
\label{Eq_SM_Incoherent_Target}
\end{align}

We prove the quantitative lower bound in Eq.~\eqref{Eq_SM_Incoherent_Target}.
Setting $\eta=0$ gives the lower bound in Theorem~1, because the conditional-min-entropy restriction at $\eta=0$ contains all resources satisfying the reduction criterion.
Combining this lower bound with the known entanglement-free upper bound~\cite{PhysRevLett.132.180805,PRXQuantum.6.020323} yields the tight scaling in Theorem~1.

We next recall the protocol model.
An incoherent adaptive protocol uses the unknown $n$-qubit Pauli channel once in each of $N$ rounds.
In round $\ell$, conditioned on the previous measurement outcomes $\vb{o}_{<\ell}=(o_1,\ldots,o_{\ell-1})$, the learner prepares an input state $\rho_{SA}^{\vb{o}_{<\ell}}$, applies the unknown channel to the system register $S$, and measures the joint system--ancilla output using a POVM
$\{E_{o_\ell}^{\vb{o}_{<\ell}}\}_{o_\ell}$.
The input state and POVM may depend arbitrarily on $\vb{o}_{<\ell}$.
However, no quantum memory is retained across distinct channel uses, so the adaptation between rounds is mediated only by the classical measurement outcomes.
The system register has dimension $d_S=2^n$, while the ancilla register $A$ may have arbitrary dimension.

We consider two one-sided restrictions.
In the input-restricted model, every state used by the learner satisfies
\begin{align}
H_{\min}(S\mid A)_{\rho_{SA}^{\vb{o}_{<\ell}}}\geq-\eta,
\end{align}
while the POVMs are unrestricted.
In the measurement-restricted model, the input states are unrestricted, while every POVM satisfies the corresponding conditional-min-entropy restriction defined in the main text.

The proof proceeds in three steps.
First, we reduce Pauli-channel learning to a partially revealed hypothesis-testing problem.
The alternative channel differs from the completely depolarizing channel in a uniformly random nonidentity Pauli direction.
This direction is revealed only after all channel uses have been completed.
Consequently, any learner that estimates all Pauli eigenvalues simultaneously must also solve the resulting testing problem.

Second, we average the alternative hypothesis over an unknown random sign.
This sign averaging cancels the contribution that is linear in the perturbation strength.
The total variation distance (TVD) between the null and alternative outcome distributions can then be bounded by a sum of roundwise squared response functions.

Finally, we use the one-sided conditional-min-entropy restriction.
For either a restricted input state or a restricted POVM, Walsh--Hadamard orthogonality bounds the average squared response to a random Pauli direction by
$O(2^{\eta-n})$ in every round, even after conditioning on all previous outcomes.
Summing this bound over the $N$ adaptive rounds gives
\begin{align}
\mathds{E}_{\vb{e}\neq0}
\TVD(Q,R_{\vb{e}})
=
O\left(
N\epsilon^2 2^{\eta-n}
\right).
\end{align}
A successful testing strategy requires this total variation distance to be bounded below by a constant, which yields Eq.~\eqref{Eq_SM_Incoherent_Target}.

\subsection{Pauli-channel preliminaries}

We use a binary labeling of Pauli operators that makes the Fourier structure of Pauli channels explicit.
For $\vb{x}=(\vb{u},\vb{v})\in\mathbb{Z}_2^{2n}$, where
$\vb{u},\vb{v}\in\mathbb{Z}_2^n$, define
\begin{align}
P_{\vb{x}}
:=
\bigotimes_{j=1}^n
i^{u_jv_j}X^{u_j}Z^{v_j}.
\end{align}
This convention gives the usual Hermitian Pauli operators and places them in one-to-one correspondence with the binary strings
$\vb{x}\in\mathbb{Z}_2^{2n}$.

For
$\vb{x}=(\vb{u},\vb{v})$ and
$\vb{b}=(\vb{u}',\vb{v}')$, define the symplectic inner product
\begin{align}
\langle\vb{x},\vb{b}\rangle
:=
\vb{u}\cdot\vb{v}'
+
\vb{v}\cdot\vb{u}'
\pmod 2.
\end{align}
It records the commutation relation between the corresponding Pauli operators:
\begin{align}
P_{\vb{x}}P_{\vb{b}}P_{\vb{x}}
=
(-1)^{\langle\vb{x},\vb{b}\rangle}
P_{\vb{b}}.
\end{align}

An $n$-qubit Pauli channel is a probabilistic mixture of Pauli conjugations,
\begin{align}
\Lambda(\rho)
=
\sum_{\vb{x}\in\mathbb{Z}_2^{2n}}
p_{\vb{x}}P_{\vb{x}}\rho P_{\vb{x}},
\end{align}
where $p_{\vb{x}}\geq0$ and
$\sum_{\vb{x}}p_{\vb{x}}=1$.
The coefficients $p_{\vb{x}}$ are the Pauli error probabilities.
The preceding conjugation relation shows that every Pauli operator is an eigenoperator of the channel:
\begin{align}
\Lambda(P_{\vb{b}})
=
\lambda_{\vb{b}}P_{\vb{b}}.
\end{align}
Its eigenvalue is
\begin{align}
\lambda_{\vb{b}}
=
\sum_{\vb{x}\in\mathbb{Z}_2^{2n}}
p_{\vb{x}}
(-1)^{\langle\vb{x},\vb{b}\rangle}.
\end{align}
Conversely,
\begin{align}
p_{\vb{x}}
=
\frac{1}{4^n}
\sum_{\vb{b}\in\mathbb{Z}_2^{2n}}
\lambda_{\vb{b}}
(-1)^{\langle\vb{x},\vb{b}\rangle}.
\end{align}
Thus, the Pauli error probabilities and Pauli eigenvalues are equivalent descriptions of the channel, related by the symplectic Walsh--Hadamard transform.
The eigenvalue $\lambda_{\vb{b}}$ corresponds to the notation $\lambda_P$ used in the main text.

In particular, the completely depolarizing channel has the uniform error distribution
$p_{\vb{x}}=4^{-n}$ for all $\vb{x}\in \mathbb{Z}_{2}^{2n}$ and eigenvalues
$\lambda_{\vb{b}}=\delta_{\vb{b},0}$.
It will serve as the null hypothesis in the testing reduction below.

\subsection{Hypothesis-testing reduction}

We now construct the testing problem used in the lower bound.
The null hypothesis is the completely depolarizing channel.
Under the alternative hypothesis, a single nontrivial Pauli eigenvalue is shifted by magnitude $3\epsilon$ in a uniformly random direction $\vb{e}$ and with a uniformly random sign $s$.
The random direction prevents the learner from tailoring its measurements to a predetermined Pauli observable, while revealing $\vb{e}$ only after the interaction converts the simultaneous-estimation requirement of Pauli-channel learning into a testing requirement.
The unrevealed sign will be averaged out below; this removes the first-order response to the perturbation and leads to the quadratic dependence on $\epsilon$.

Assume throughout that $0<\epsilon\leq1/3$.
For $\vb{e}\in\mathbb{Z}_2^{2n}\setminus\{0^{2n}\}$ and
$s\in\{\pm1\}$, define
\begin{align}
H_0:\quad
&\Lambda_{\mathrm{dep}},
\qquad
\lambda_{\vb{b}}
=
\delta_{\vb{b},0},
\\
H_1:\quad
&\Lambda_{\vb{e},s},
\qquad
\lambda_{\vb{b}}
=
\delta_{\vb{b},0}
+
3s\epsilon\,\delta_{\vb{b},\vb{e}}.
\label{Eq_SM_Hypothesis}
\end{align}

The corresponding partially revealed testing game is defined as follows.

\begin{enumerate}
    \item The referee samples a nonzero bit string
    $\vb{e}\in\mathbb{Z}_2^{2n}\setminus\{0^{2n}\}$ and a sign
    $s\in\{-1,1\}$ uniformly at random.

        \item The referee selects $H_0$ or $H_1$ from Eq.~\eqref{Eq_SM_Hypothesis} with equal probability and gives the player $N$ uses of the chosen channel.

    \item The player performs an incoherent adaptive protocol.
    In each channel use, the input state and POVM may depend on the previous measurement outcomes.
    We consider two one-sided protocol classes: (i) every input state satisfies
    $H_{\min}(S\mid A)\geq-\eta$, while the POVMs are unrestricted; or
    (ii) every POVM satisfies the corresponding conditional-min-entropy restriction, while the input states are unrestricted.
    All channel uses are completed before the referee reveals $\vb{e}$.

    \item Finally, the referee reveals the bit string $\vb{e}$, but not the
    sign $s$, and asks the player to decide whether the selected channel was
    $\Lambda_{\mathrm{dep}}$.
\end{enumerate}

Any Pauli-channel learner induces a strategy for this hypothesis-testing game.
After $\vb{e}$ is revealed, the player uses the learner's estimate $\hat{\lambda}_{\vb{e}}$.
Under the learner's estimation guarantee, $H_0$ implies $|\hat{\lambda}_{\vb{e}}|\leq\epsilon$, whereas $H_1$ implies $|\hat{\lambda}_{\vb{e}}|\geq2\epsilon$.
Therefore, thresholding $|\hat{\lambda}_{\vb{e}}|$ at $3\epsilon/2$ distinguishes the two hypotheses.
A learner with success probability at least $2/3$ thus yields a testing strategy with the same success probability.
Consequently, any lower bound for the testing game also applies to the learning task.

Let $Q(\vb{o}_{1:N})$ denote the distribution of the full outcome $\vb{o}_{1:N}=(o_1,\ldots,o_N)$ under $H_0$, and let $R_{\vb{e},s}(\vb{o}_{1:N})$ denote the corresponding distribution under $H_1$ with spike location $\vb{e}$ and sign $s$.
Since the sign $s$ is not revealed to the player, the relevant alternative distribution for a fixed $\vb{e}$ is the sign-averaged distribution
\begin{align}
    R_{\vb{e}}(\vb{o}_{1:N})
    :=
    \mathds{E}_{s}R_{\vb{e},s}(\vb{o}_{1:N}).
\end{align}

Le Cam's two-point method~\cite{lecam1973convergence} bounds the player's success probability in terms of the TVD between the null distribution and the sign-averaged alternative distribution.
Averaging this bound over the revealed nonzero index $\vb{e}$ gives
\begin{align}
\Pr (\text{win}) &\leq \frac{1}{2} \left( 1+
        \mathds{E}_{\vb{e}\neq 0}
        \TVD\!\left(
            Q(\vb{o}_{1:N}),
            R_{\vb{e}}(\vb{o}_{1:N})
        \right)\right).
\end{align}
Here, TVD is defined as
\begin{align}
    \TVD(P,Q)
    :=
    \sum_{\omega}\max\{0,P(\omega)-Q(\omega)\}
    =
    \frac{1}{2}\sum_{\omega}|P(\omega)-Q(\omega)|.
    \label{Eq_SM_TVD_Def}
\end{align}
Combining this with $\Pr(\mathrm{win})\ge 2/3$ gives
\begin{align}
    \frac{1}{3}
    \leq
    \mathds{E}_{\vb{e}\neq 0}
    \TVD\!\left(
        Q(\vb{o}_{1:N}),
        R_{\vb{e}}(\vb{o}_{1:N})
    \right).
    \label{Eq_SM_TVD_Lower}
\end{align}

Applying the inverse Walsh--Hadamard transform to the Pauli eigenvalues in Eq.~\eqref{Eq_SM_Hypothesis} gives
\begin{align}
    p_{\mathrm{dep}}(\vb{x})
    =
    \frac{1}{4^n},
    \qquad
    p_{\vb{e},s}(\vb{x})
    =
    \frac{1}{4^n}
    \left[
        1+
        3s\epsilon
        (-1)^{\langle\vb{x},\vb{e}\rangle}
    \right].
\end{align}
Because $\vb{e}\neq0$, the values $+1$ and $-1$ occur equally often in $(-1)^{\langle\vb{x},\vb{e}\rangle}$ as $\vb{x}$ ranges over $\mathbb{Z}_2^{2n}$; hence, $\sum_{\vb{x}\in\mathbb{Z}_2^{2n}}p_{\vb{e},s}(\vb{x})=1$.
The assumption $0<\epsilon\leq1/3$ also ensures that $p_{\vb{e},s}(\vb{x})\geq0$ for every $\vb{x}$.
Therefore, $\Lambda_{\vb{e},s}$ is a valid completely positive and trace-preserving (CPTP) Pauli channel.
Explicitly,
\begin{align}
    \Lambda_{\vb{e},s}(\rho)
    =
    \sum_{\vb{x}\in\mathbb{Z}_2^{2n}}
    \frac{
        1+
        3s\epsilon
        (-1)^{\langle\vb{x},\vb{e}\rangle}
    }{4^n}
    P_{\vb{x}}\rho P_{\vb{x}}.
\end{align}

For each round $\ell$, let $R_{\vb{e},s}(o_{\ell}\mid\vb{o}_{<\ell})$ and $Q(o_{\ell}\mid\vb{o}_{<\ell})$ denote the conditional probabilities of obtaining outcome $o_{\ell}$ under $H_1$ and $H_0$, respectively, given the previous measurement outcomes $\vb{o}_{<\ell}=(o_1,\ldots,o_{\ell-1})$.
For the input state $\rho_{SA}^{\vb{o}_{<\ell}}$ and POVM $\{E_{o_{\ell}}^{\vb{o}_{<\ell}}\}_{o_{\ell}}$ selected according to these previous measurement outcomes, the conditional probabilities are
\begin{align}
    R_{\vb{e},s}(o_{\ell}\mid\vb{o}_{<\ell})
    &=
    \Tr\left[
        E_{o_{\ell}}^{\vb{o}_{<\ell}}
        \left(
            \Lambda_{\vb{e},s}\otimes\id_A
        \right)
        \left(
            \rho_{SA}^{\vb{o}_{<\ell}}
        \right)
    \right], \qquad
    Q(o_{\ell}\mid\vb{o}_{<\ell})
    =
    \Tr\left[
        E_{o_{\ell}}^{\vb{o}_{<\ell}}
        \left(
            \Lambda_{\mathrm{dep}}\otimes\id_A
        \right)
        \left(
            \rho_{SA}^{\vb{o}_{<\ell}}
        \right)
    \right].
\end{align}
Using the explicit form of $\Lambda_{\vb{e},s}$, we write
\begin{align}
    m_{\vb{e}}^{\vb{o}_{<\ell}}(o_{\ell})
    &:=
    \sum_{\vb{x}\in\mathbb{Z}_2^{2n}}
    \frac{(-1)^{\langle\vb{x},\vb{e}\rangle}}{4^n}
    \Tr\left[
        E_{o_{\ell}}^{\vb{o}_{<\ell}}
        \left(P_{\vb{x}}\otimes I_A\right)
        \rho_{SA}^{\vb{o}_{<\ell}}
        \left(P_{\vb{x}}\otimes I_A\right)
    \right],
    \\
    R_{\vb{e},s}(o_{\ell}\mid\vb{o}_{<\ell})
    &=
    Q(o_{\ell}\mid\vb{o}_{<\ell})
    +
    3s\epsilon\,m_{\vb{e}}^{\vb{o}_{<\ell}}(o_{\ell}).
\end{align}
For brevity, write $\vb{o}:=\vb{o}_{1:N}$.
For a fixed nonzero $\vb{e}$, Eq.~\eqref{Eq_SM_TVD_Def} gives
\begin{align}
    \TVD(Q,R_{\vb{e}})
    =
    \sum_{\vb{o}}
    \max\!\left\{
        0,
        Q(\vb{o})-R_{\vb{e}}(\vb{o})
    \right\}.
\end{align}
Measurement-outcome sequences satisfying $Q(\vb{o})=0$ do not contribute to this sum.
For any measurement-outcome sequence satisfying $Q(\vb{o})>0$, the definition of $R_{\vb{e}}$ and the chain rule yield
\begin{align}
    Q(\vb{o})-R_{\vb{e}}(\vb{o})
    &=
    Q(\vb{o})
    \left[
        1-
        \mathds{E}_{s}
        \frac{R_{\vb{e},s}(\vb{o})}{Q(\vb{o})}
    \right]
    =
    Q(\vb{o})
    \left[
        1-
        \mathds{E}_{s}
        \prod_{\ell=1}^{N}
        \frac{R_{\vb{e},s}(o_{\ell}\mid\vb{o}_{<\ell})}{Q(o_{\ell}\mid\vb{o}_{<\ell})}
    \right]
    =
    Q(\vb{o})
    \left[
        1-
        \mathds{E}_{s}
        \prod_{\ell=1}^{N}
        \left(
            1+
            3s\epsilon
            \frac{m_{\vb{e}}^{\vb{o}_{<\ell}}(o_{\ell})}{Q(o_{\ell}\mid\vb{o}_{<\ell})}
        \right)
    \right].
\end{align}
For $Q(o_{\ell}\mid\vb{o}_{<\ell})>0$, define
\begin{align}
    G_{\vb{e}}^{\vb{o}_{\leq\ell}}
    :=
    \frac{
        m_{\vb{e}}^{\vb{o}_{<\ell}}(o_{\ell})
    }{
        Q(o_{\ell}\mid\vb{o}_{<\ell})
    }.
\end{align}
When $Q(o_{\ell}\mid\vb{o}_{<\ell})=0$, set $G_{\vb{e}}^{\vb{o}_{\leq\ell}}=0$.
In this case, the equality $R_{\vb{e},s}(o_{\ell}\mid\vb{o}_{<\ell})=Q(o_{\ell}\mid\vb{o}_{<\ell})+3s\epsilon\,m_{\vb{e}}^{\vb{o}_{<\ell}}(o_{\ell})$ and the nonnegativity of $R_{\vb{e},s}(o_{\ell}\mid\vb{o}_{<\ell})$ for both signs imply $m_{\vb{e}}^{\vb{o}_{<\ell}}(o_{\ell})=0$.
This convention is therefore consistent.
For $Q(o_{\ell}\mid\vb{o}_{<\ell})>0$, we have
\begin{align}
    1+
    3s\epsilon
    G_{\vb{e}}^{\vb{o}_{\leq\ell}}
    =
    \frac{
        R_{\vb{e},s}(o_{\ell}\mid\vb{o}_{<\ell})
    }{
        Q(o_{\ell}\mid\vb{o}_{<\ell})
    }
    \geq
    0,
    \qquad
    s\in\{\pm1\}.
\end{align}
It follows that $\left|3\epsilon G_{\vb{e}}^{\vb{o}_{\leq\ell}}\right|\leq1$.
The arithmetic--geometric mean inequality therefore gives
\begin{align}
    1-
    \mathds{E}_{s}
    \prod_{\ell=1}^{N}
    \left(
        1+
        3s\epsilon
        G_{\vb{e}}^{\vb{o}_{\leq\ell}}
    \right)
    &\leq
    1-
    \sqrt{
        \prod_{\ell=1}^{N}
        \left(
            1-
            3\epsilon
            G_{\vb{e}}^{\vb{o}_{\leq\ell}}
        \right)
        \left(
            1+
            3\epsilon
            G_{\vb{e}}^{\vb{o}_{\leq\ell}}
        \right)
    }
     =
    1-
    \prod_{\ell=1}^{N}
    \sqrt{
        1-
        9\epsilon^2
        \left(
            G_{\vb{e}}^{\vb{o}_{\leq\ell}}
        \right)^2
    }
    \\
    &\leq
    1-
    \prod_{\ell=1}^{N}
    \left[
        1-
        9\epsilon^2
        \left(
            G_{\vb{e}}^{\vb{o}_{\leq\ell}}
        \right)^2
    \right]
    \leq
    9\epsilon^2
    \sum_{\ell=1}^{N}
    \left(
        G_{\vb{e}}^{\vb{o}_{\leq\ell}}
    \right)^2.
\end{align}
The last two inequalities follow from $\sqrt{1-x}\geq1-x$ and $1-\prod_{\ell=1}^{N}(1-x_{\ell})\leq\sum_{\ell=1}^{N}x_{\ell}$ for $x,x_{\ell}\in[0,1]$.
Combining this estimate with the preceding likelihood-ratio expression gives
\begin{align}
    Q(\vb{o})-R_{\vb{e}}(\vb{o})
    &\leq
    9\epsilon^2
    Q(\vb{o})
    \sum_{\ell=1}^{N}
    \left(
        G_{\vb{e}}^{\vb{o}_{\leq\ell}}
    \right)^2.
\end{align}
Under the convention established above, the quotient below is set to zero whenever $Q(o_{\ell}\mid\vb{o}_{<\ell})=0$.
Summing the pointwise bound over all measurement-outcome sequences and then averaging over nonzero $\vb{e}$ gives
\begin{align}
    \mathds{E}_{\vb{e}\neq0}\TVD(Q,R_{\vb{e}})
    &\leq
    9\epsilon^2
    \sum_{\ell=1}^{N}
    \sum_{\vb{o}_{<\ell}}
    Q(\vb{o}_{<\ell})
    \mathds{E}_{\vb{e}\neq0}\!\left[
        \sum_{o_{\ell}}
        \frac{\left(m_{\vb{e}}^{\vb{o}_{<\ell}}(o_{\ell})\right)^2}{Q(o_{\ell}\mid\vb{o}_{<\ell})}
    \right].
    \label{Eq_SM_TVD_Response_Bound}
\end{align}
Equation~\eqref{Eq_SM_TVD_Response_Bound} reduces the TVD to a single-round estimate. It remains to bound the average squared response for every conditional input state and POVM selected along the adaptive protocol. 
We establish this estimate separately under the two one-sided restrictions.

For fixed previous measurement outcomes $\vb{o}_{<\ell}$, define
\begin{align}
    a_{\vb{x}}^{\vb{o}_{<\ell}}(o_{\ell})
    &:=
    \Tr\left[
        E_{o_{\ell}}^{\vb{o}_{<\ell}}
        \left(P_{\vb{x}}\otimes I_A\right)
        \rho_{SA}^{\vb{o}_{<\ell}}
        \left(P_{\vb{x}}\otimes I_A\right)
    \right],
    \\
    m_{\vb{e}}^{\vb{o}_{<\ell}}(o_{\ell})
    &=
    \frac{1}{4^n}
    \sum_{\vb{x}\in\mathbb{Z}_2^{2n}}
    (-1)^{\langle\vb{x},\vb{e}\rangle}
    a_{\vb{x}}^{\vb{o}_{<\ell}}(o_{\ell}).
\end{align}
We bound this expression separately under the input-state and POVM restrictions.

\subsection{Input-state restriction}
Assume that every input state satisfies
\begin{align}
    H_{\min}(S\mid A)_{\rho_{SA}^{\vb{o}_{<\ell}}}
    \geq
    -\eta.
\end{align}
By the definition of conditional min-entropy, for each set of previous measurement outcomes there exists a density operator $\sigma_A^{\vb{o}_{<\ell}}$ such that
\begin{align}
    \rho_{SA}^{\vb{o}_{<\ell}}
    \leq
    2^\eta\left(I_S\otimes\sigma_A^{\vb{o}_{<\ell}}\right).
\end{align}
When $\eta=0$, this condition includes every state satisfying the reduction criterion and hence every bound-entangled state.
For fixed previous measurement outcomes, define
\begin{align}
    t_{o_{\ell}}
    :=
    2^\eta
    \Tr\left[
        E_{o_{\ell}}^{\vb{o}_{<\ell}}
        \left(I_S\otimes\sigma_A^{\vb{o}_{<\ell}}\right)
    \right].
\end{align}
The state inequality implies, for every $\vb{x}\in\mathbb{Z}_2^{2n}$,
\begin{align}
    a_{\vb{x}}^{\vb{o}_{<\ell}}(o_{\ell})
    \leq
    t_{o_{\ell}}.
    \label{Eq_SM_state}
\end{align}
Let $\mathds{E}_{\vb{e}}$ denote the uniform average over all
$\vb{e}\in\mathbb{Z}_2^{2n}$, including $\vb{e}=0$.
Using $\mathds{E}_{\vb{e}}\left[(-1)^{\langle\vb{x}\oplus\vb{x}^{\prime},\vb{e}\rangle}\right]=\delta_{\vb{x},\vb{x}^{\prime}}$, together with $a_{\vb{x}}^{\vb{o}_{<\ell}}(o_{\ell})\leq t_{o_{\ell}}$, we obtain
\begin{align}
    \mathds{E}_{\vb{e}}\!\left[
        \left(
            m_{\vb{e}}^{\vb{o}_{<\ell}}(o_{\ell})
        \right)^2
    \right]
    &=
    \frac{1}{4^{2n}}
    \sum_{\vb{x}\in\mathbb{Z}_2^{2n}}
    \left(
        a_{\vb{x}}^{\vb{o}_{<\ell}}(o_{\ell})
    \right)^2 \leq \frac{1}{4^{2n}}
    \sum_{\vb{x}\in\mathbb{Z}_2^{2n}}
        a_{\vb{x}}^{\vb{o}_{<\ell}}(o_{\ell}) t_{o_{\ell}}
    \leq
    \frac{t_{o_{\ell}}}{4^n}
    Q(o_{\ell}\mid\vb{o}_{<\ell}).
\end{align}
Since the contribution from $\vb{e}=0$ is nonnegative, the zero-denominator convention gives
\begin{align}
    \mathds{E}_{\vb{e}\neq0}\!\left[
        \sum_{o_{\ell}}
        \frac{
            \left(
                m_{\vb{e}}^{\vb{o}_{<\ell}}(o_{\ell})
            \right)^2
        }{
            Q(o_{\ell}\mid\vb{o}_{<\ell})
        }
    \right]
    &\leq
    \frac{4^n}{4^n-1}
    \frac{1}{4^n}
    \sum_{o_{\ell}}t_{o_{\ell}}
    =
    \frac{4^n}{4^n-1}
    2^{\eta-n},
\end{align}
where the last equality follows from the completeness of the POVM:
\begin{align}
    \sum_{o_{\ell}}t_{o_{\ell}}
    =
    2^\eta
    \Tr\left[
        I_{SA}
        \left(
            I_S\otimes\sigma_A^{\vb{o}_{<\ell}}
        \right)
    \right]
    =
    2^{n+\eta}.
\end{align}
Substituting this estimate into Eq.~\eqref{Eq_SM_TVD_Response_Bound} and using $\sum_{\vb{o}_{<\ell}}Q(\vb{o}_{<\ell})=1$ gives
\begin{align}
    \mathds{E}_{\vb{e}\neq0}\TVD(Q,R_{\vb{e}})
    \leq
    9\epsilon^2N\frac{4^n}{4^n-1}2^{\eta-n}
    \leq
    12\epsilon^2N2^{\eta-n}.
\end{align}
Combining this upper bound with Eq.~\eqref{Eq_SM_TVD_Lower} yields
\begin{align}
    N
    =
    \Omega\!\left(\epsilon^{-2}2^{n-\eta}\right).
\end{align}
This proves the input-state-restricted case of Theorem~1 and the corresponding incoherent-adaptive case of Corollary~1.

\subsection{POVM restriction}
Assume that for each set of previous measurement outcomes there exist density operators $\{\sigma_{o_{\ell}}^{\vb{o}_{<\ell}}\}_{o_{\ell}}$ satisfying
\begin{align}
    \frac{E_{o_{\ell}}^{\vb{o}_{<\ell}}}{\Tr E_{o_{\ell}}^{\vb{o}_{<\ell}}}
    &\leq
    2^\eta\left(I_S\otimes\sigma_{o_{\ell}}^{\vb{o}_{<\ell}}\right),\qquad\sum_{o_{\ell}}\Tr\left(E_{o_{\ell}}^{\vb{o}_{<\ell}}\right)\sigma_{o_{\ell}}^{\vb{o}_{<\ell}}
    =
    2^n I_A.
\end{align}
The first inequality is imposed for every outcome with nonzero trace; an effect with zero trace vanishes and does not contribute.

For fixed previous measurement outcomes, let
$\rho_A^{\vb{o}_{<\ell}}:=\Tr_S\rho_{SA}^{\vb{o}_{<\ell}}$ and define
\begin{align}
t_{o_{\ell}}
:=
2^\eta
\Tr\left(E_{o_{\ell}}^{\vb{o}_{<\ell}}\right)
\Tr\left[
\sigma_{o_{\ell}}^{\vb{o}_{<\ell}}
\rho_A^{\vb{o}_{<\ell}}
\right].
\end{align}
Since $P_{\vb{x}}$ acts only on $S$, the POVM inequality gives
$a_{\vb{x}}^{\vb{o}_{<\ell}}(o_{\ell})\leq t_{o_{\ell}}$ for every
$\vb{x}\in\mathbb{Z}_2^{2n}$.
Using the averaging identity from the preceding subsection, we obtain
\begin{align}
\mathds{E}_{\vb{e}}\left[
\left(
m_{\vb{e}}^{\vb{o}_{<\ell}}(o_{\ell})
\right)^2
\right]
=
\frac{1}{4^{2n}}
\sum_{\vb{x}\in\mathbb{Z}_2^{2n}}
\left(
a_{\vb{x}}^{\vb{o}_{<\ell}}(o_{\ell})
\right)^2
\leq
\frac{t_{o_{\ell}}}{4^n}
Q(o_{\ell}\mid\vb{o}_{<\ell}).
\end{align}
Therefore, using the zero-denominator convention and conditioning on
$\vb{e}\neq0$,
\begin{align}
\mathds{E}_{\vb{e}\neq0}\left[
\sum_{o_{\ell}}
\frac{
\left(
m_{\vb{e}}^{\vb{o}_{<\ell}}(o_{\ell})
\right)^2
}{
Q(o_{\ell}\mid\vb{o}_{<\ell})
}
\right]
\leq
\frac{4^n}{4^n-1}
\frac{1}{4^n}
\sum_{o_{\ell}}t_{o_{\ell}}
=
\frac{4^n}{4^n-1}2^{\eta-n},
\end{align}
where
\begin{align}
\sum_{o_{\ell}}t_{o_{\ell}}
=
2^\eta
\Tr\left[
\left(
\sum_{o_{\ell}}
\Tr\left(E_{o_{\ell}}^{\vb{o}_{<\ell}}\right)
\sigma_{o_{\ell}}^{\vb{o}_{<\ell}}
\right)
\rho_A^{\vb{o}_{<\ell}}
\right]
=
2^{\eta+n}.
\end{align}
Substituting this estimate into Eq.~\eqref{Eq_SM_TVD_Response_Bound} gives
\begin{align}
    \mathds{E}_{\vb{e}\neq0}\TVD(Q,R_{\vb{e}})
    \leq
    9\epsilon^2N\frac{4^n}{4^n-1}2^{\eta-n}
    \leq
    12\epsilon^2N2^{\eta-n}.
\end{align}
Combining this upper bound with Eq.~\eqref{Eq_SM_TVD_Lower} yields
\begin{align}
    N
    =
    \Omega\!\left(\epsilon^{-2}2^{n-\eta}\right).
\end{align}
This proves the POVM-restricted case of Theorem~1 and the corresponding incoherent-adaptive case of Corollary~1.

\section{Bounds for one-sided coherent adaptive Pauli-channel learning}

In this section, we prove Theorem~2 and the coherent-adaptive part of Corollary~1 of the main text.
We first recall the precise statements proved in this section.
\medskip

\noindent\textbf{Theorem 2 (main text).}
Assume $0<\epsilon\leq1/3$.
Consider either of the two one-sided coherent adaptive models.
In the input-restricted model, each round uses a fresh reduction-criterion-restricted input that is uncorrelated with the retained quantum system conditioned on the previous measurement outcomes, while arbitrary coherent processing of the retained quantum systems is allowed.
In the measurement-restricted model, the inputs may be entangled across channel uses, while in each round a reduction-criterion-restricted POVM acts only on the current channel output and its associated ancilla.
In either model, Pauli-channel learning requires
\begin{align}
    N
    =
    \Omega\left(
        \epsilon^{-1}2^{n/2}
    \right).
\end{align}

\medskip
\noindent\textbf{Corollary 1 (coherent-adaptive part; main text).}
Assume $0<\epsilon\leq1/3$.
Suppose that the restricted input states or normalized measurement effects satisfy
$H_{\min}(S\mid A)\geq-\eta$ for some fixed $\eta \geq 0$.
Then either one-sided coherent adaptive model requires
\begin{align}
    N
    =
    \Omega\left(
        \epsilon^{-1}2^{(n-\eta)/2}
    \right).
    \label{Eq_SM_Coherent_Target}
\end{align}

We prove the quantitative lower bound in Eq.~\eqref{Eq_SM_Coherent_Target}.
Setting $\eta=0$ gives Theorem~2 because the conditional-min-entropy restriction at $\eta=0$ contains all resources satisfying the reduction criterion.

The proof uses the same randomly spiked channel family and learning-to-testing reduction as in Section~\ref{SM_Section_1}.
Unlike the incoherent-adaptive proof, however, we cannot decompose the distinguishability into roundwise squared classical responses because quantum information may be retained across rounds.
Instead, we interpolate between the null and alternative protocols by replacing one channel use at a time.
Trace-norm contractivity ensures that subsequent coherent processing cannot increase the distinguishability introduced at the replaced channel use.
The one-sided conditional-min-entropy restriction, together with the leftover-hash lemma proved below, bounds the average distinguishability introduced by each replacement by
$O\left(\epsilon 2^{-(n-\eta)/2}\right)$.
Summing over the $N$ channel uses gives Eq.~\eqref{Eq_SM_Coherent_Target}.

We consider the input-restricted and measurement-restricted models separately because the quantum systems retained between channel uses have different forms.
The same hybrid argument will apply once the corresponding branch states have been defined.

The learning-to-testing reduction of Section~\ref{SM_Section_1} is independent of whether quantum memory is retained.
After including any final measurement in the protocol, it applies to the complete classical outcome record in both coherent models.
Let $\vb{o}_{\leq N}$ denote this record, and let
$Q(\vb{o}_{\leq N})$ and $R_{\vb{e}}(\vb{o}_{\leq N})$
denote its distributions under the null and sign-averaged alternative hypotheses, respectively.
Consequently,
\begin{align}
    \frac{1}{3}\leq\mathds{E}_{\vb{e}\neq0}\TVD\!\left(Q(\vb{o}_{\leq N}),R_{\vb{e}}(\vb{o}_{\leq N})\right).
    \label{Eq_SM_Coherent_TVD_Lower}
\end{align}

For $m=0,\ldots,N$, define the hybrid channel sequence
\begin{align}
    \vec{\Lambda}_{\vb{e},m}^{s}
    &:=
    \left(
        \underbrace{
            \Lambda_{\dep},\ldots,\Lambda_{\dep}
        }_{m\text{ copies}},
        \underbrace{
            \Lambda_{\vb{e},s},\ldots,\Lambda_{\vb{e},s}
        }_{N-m\text{ copies}}
    \right),
    \label{Eq_SM_Coherent_Hybrid_Sequence}
\end{align}
where $m=N$ corresponds to the null hypothesis, while $m=0$
corresponds to the alternative hypothesis with fixed $\vb{e}$ and
$s$.

\subsection{Input-restricted coherent protocols}
\subsubsection{Model description}
Let $\vec{\Lambda}:=(\Lambda_1,\ldots,\Lambda_N)$ denote the ordered channel sequence used to describe the protocol.
In the null and alternative hypotheses, all components of $\vec{\Lambda}$ are identical, whereas the hybrid sequences in Eq.~\eqref{Eq_SM_Coherent_Hybrid_Sequence} may contain different channels.
For each round $\ell\in\{1,\ldots,N\}$, let $S_{\ell}$ denote the $n$-qubit system register on which $\Lambda_{\ell}$ acts, let $A_{\ell}$ denote an associated ancilla register of arbitrary dimension, and let $W_{\ell-1}$ denote the quantum system retained after the preceding rounds.
Conditioned on the previous measurement outcomes $\vb{o}_{<\ell}$, the learner prepares a fresh input state $\rho_{\ell}^{\vb{o}_{<\ell}}:=\rho_{S_{\ell}A_{\ell}}^{\vb{o}_{<\ell}}$ satisfying
\begin{align}
    H_{\min}(S_{\ell}\mid A_{\ell})_{\rho_{\ell}^{\vb{o}_{<\ell}}}\geq-\eta.
\end{align}
Conditioned on $\vb{o}_{<\ell}$, the fresh input state is uncorrelated with $W_{\ell-1}$.
After $\Lambda_{\ell}$ acts on $S_{\ell}$, the learner may apply an arbitrary quantum instrument to $W_{\ell-1}S_{\ell}A_{\ell}$, produce the measurement outcome $o_{\ell}$, and retain an arbitrary quantum system $W_{\ell}$.
We call the point immediately after round $\ell$ and before the next input is prepared the $\ell$th causal cut.
Write $\vb{o}_{\leq\ell}:=(o_1,\ldots,o_{\ell})=\vb{o}_{<\ell+1}$.
At the $\ell$th causal cut, let $\Gamma_{\vb{o}_{\leq\ell}}^{\vec{\Lambda},\ell}$ denote the subnormalized state of $W_{\ell}$ associated with the measurement outcomes $\vb{o}_{\leq\ell}$.
The corresponding classical--quantum state is defined as follows:
\begin{align}
    \Omega_{\ell}^{\vec{\Lambda}}:=\sum_{\vb{o}_{\leq\ell}}\ketbra{\vb{o}_{\leq\ell}}{\vb{o}_{\leq\ell}}\otimes\Gamma_{\vb{o}_{\leq\ell}}^{\vec{\Lambda},\ell}.
\end{align}
The probability of each branch is obtained from its trace:
\begin{align}
    \Tr\left[\Gamma_{\vb{o}_{\leq\ell}}^{\vec{\Lambda},\ell}\right]=\Pr_{\vec{\Lambda}}\!\left[\vb{o}_{\leq\ell}\right].
\end{align}
We take $W_0$ to be one-dimensional and set $\Gamma_{\emptyset}^{\vec{\Lambda},0}=1$.

\begin{figure}[t]
    \centering
    \includegraphics[width=0.97\textwidth]{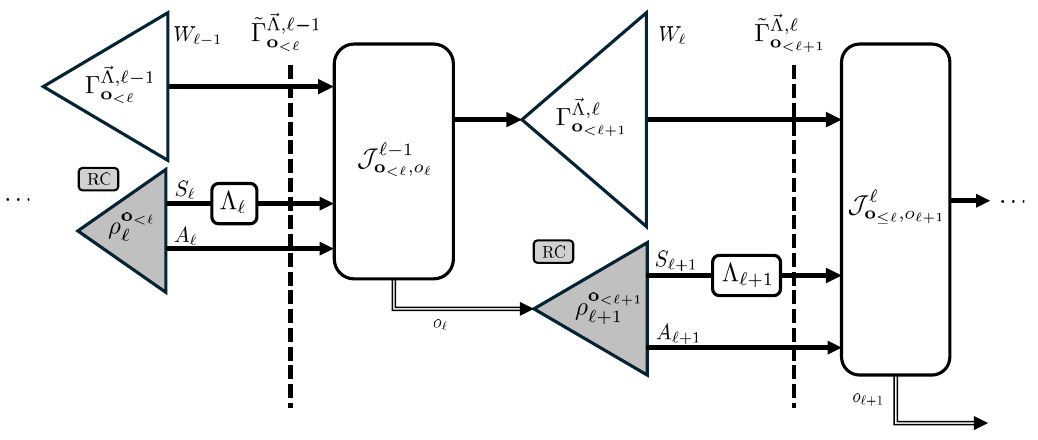}
\caption{
Input-restricted coherent adaptive protocol.
In round $\ell$, the learner prepares a fresh state $\rho_{\ell}^{\vb{o}_{<\ell}}$ on $S_{\ell}A_{\ell}$ satisfying the conditional-min-entropy restriction and sends $S_{\ell}$ through $\Lambda_{\ell}$.
An unrestricted quantum instrument $\{\mathcal{J}_{\vb{o}_{<\ell},o_{\ell}}^{\ell-1}\}_{o_{\ell}}$ acts on $W_{\ell-1}S_{\ell}A_{\ell}$, produces the measurement outcome $o_{\ell}$, and outputs the retained system $W_{\ell}$.
The next fresh input is selected according to $\vb{o}_{\leq \ell}$.
}
    \label{fig_SM_Quantum_Restricted_Input}
\end{figure}

The transition from the $(\ell-1)$th causal cut to the $\ell$th causal cut is shown in Fig.~\ref{fig_SM_Quantum_Restricted_Input}.
Before the round-$\ell$ instrument is applied, the joint state is:
\begin{align}
    \tilde{\Gamma}_{\vb{o}_{<\ell}}^{\vec{\Lambda},\ell-1}:=\left(\id_{W_{\ell-1}}\otimes\Lambda_{\ell}\otimes\id_{A_{\ell}}\right)\!\left(\Gamma_{\vb{o}_{<\ell}}^{\vec{\Lambda},\ell-1}\otimes\rho_{\ell}^{\vb{o}_{<\ell}}\right).
\end{align}
For each fixed sequence of previous measurement outcomes, the round-$\ell$ operation is described by the quantum instrument
\begin{align}
    \left\{\mathcal{J}_{\vb{o}_{<\ell},o_{\ell}}^{\ell-1}:\mathcal{L}(W_{\ell-1}S_{\ell}A_{\ell})\rightarrow\mathcal{L}(W_{\ell})\right\}_{o_{\ell}}.
\end{align}
Each branch map $\mathcal{J}_{\vb{o}_{<\ell},o_{\ell}}^{\ell-1}$ is completely positive and trace-nonincreasing (CPTNI).
It may include arbitrary post-channel processing, measurements, ancillary preparation, discarding of subsystems, and the retention of $W_{\ell}$ for subsequent rounds.
Any classical outcomes generated during round $\ell$ are included in the composite label $o_{\ell}$.
For each fixed $\vb{o}_{<\ell}$, the normalization condition is
\begin{align}
    \sum_{o_{\ell}}
    \mathcal{J}_{\vb{o}_{<\ell},o_{\ell}}^{\ell-1}
    \quad\text{is CPTP}.
\end{align}

After $o_{\ell}$ is stored in a classical register, the branch maps define the outcome-labeled CPTP map introduced in the lower-bound proof.
The updated branch state is:
\begin{align}
    \Gamma_{\vb{o}_{\leq\ell}}^{\vec{\Lambda},\ell}=\mathcal{J}_{\vb{o}_{<\ell},o_{\ell}}^{\ell-1}\!\left(\tilde{\Gamma}_{\vb{o}_{<\ell}}^{\vec{\Lambda},\ell-1}\right).
\end{align}
Any final measurement of the retained system can be included in the round-$N$ instrument, so we take $W_N$ to be one-dimensional.
The final classical--quantum state and outcome probabilities are therefore
\begin{align}
    \Omega_N^{\vec{\Lambda}}
    &=
    \sum_{\vb{o}_{\leq N}}
    \ketbra{\vb{o}_{\leq N}}{\vb{o}_{\leq N}}
    \otimes
    \Gamma_{\vb{o}_{\leq N}}^{\vec{\Lambda},N},
    \qquad
    \Pr_{\vec{\Lambda}}\left[\vb{o}_{\leq N}\right]
    =
    \Tr\left[
        \Gamma_{\vb{o}_{\leq N}}^{\vec{\Lambda},N}
    \right].
\end{align}

\subsubsection{Lower-bound proof}
For fixed nonzero $\vb{e}$ and $s\in\{\pm1\}$, let $R_{\vb{e},s}(\vb{o}_{\leq N}):=\Pr_{\vec{\Lambda}_{\vb{e},0}^{s}}\!\left[\vb{o}_{\leq N}\right]$.
The null distribution satisfies $Q(\vb{o}_{\leq N})=\Pr_{\vec{\Lambda}_{\vb{e},N}^{s}}\!\left[\vb{o}_{\leq N}\right]$, which is independent of $\vb{e}$ and $s$.
The sign-averaged alternative distribution is $R_{\vb{e}}(\vb{o}_{\leq N})=\mathds{E}_{s}R_{\vb{e},s}(\vb{o}_{\leq N})$.
By Eq.~\eqref{Eq_SM_Coherent_TVD_Lower}, it remains to upper bound $\mathds{E}_{\vb{e}\neq0}\TVD(Q,R_{\vb{e}})$.
For $m=0,\ldots,N$ and $\ell=1,\ldots,N$, use the shorthand
\begin{align}
    \Omega_{m,\ell}^{\vb{e},s}:=\Omega_{\ell}^{\vec{\Lambda}_{\vb{e},m}^{s}},
    \qquad
    \Gamma_{\vb{o}_{\leq\ell}}^{m,s,\ell}:=\Gamma_{\vb{o}_{\leq\ell}}^{\vec{\Lambda}_{\vb{e},m}^{s},\ell}.
\end{align}
Because $W_N$ is one-dimensional, for each $s\in\{\pm1\}$ the quantity $\left\|\Omega_{N,N}^{\vb{e},s}-\Omega_{0,N}^{\vb{e},s}\right\|_1$ equals the $\ell_1$ distance between $Q$ and $R_{\vb{e},s}$.
For each fixed nonzero $\vb{e}$, the triangle inequality for the sign average followed by the hybrid telescoping sum gives
\begin{align}
    \TVD(Q,R_{\vb{e}})
    &\leq
    \frac{1}{4}\sum_{s=\pm1}\left\|\Omega_{N,N}^{\vb{e},s}-\Omega_{0,N}^{\vb{e},s}\right\|_1
    \leq
    \frac{1}{4}\sum_{s=\pm1}\sum_{m=1}^{N}\left\|\Omega_{m,N}^{\vb{e},s}-\Omega_{m-1,N}^{\vb{e},s}\right\|_1.
\end{align}

We first show that rounds after the differing channel use cannot increase the trace-norm difference.
For fixed previous measurement outcomes $\vb{o}_{<\ell}$, define the outcome-labeled map
\begin{align}
    \mathcal{M}_{\vb{o}_{<\ell}}^{\ell}(X):=\sum_{o_{\ell}}\ketbra{\vb{o}_{\leq\ell}}{\vb{o}_{\leq\ell}}\otimes\mathcal{J}_{\vb{o}_{<\ell},o_{\ell}}^{\ell-1}(X).
\end{align}
By the instrument normalization condition, $\mathcal{M}_{\vb{o}_{<\ell}}^{\ell}$ is CPTP.
Fix $m<\ell$ and define
\begin{align}
    \Delta_{\vb{o}_{<\ell}}^{m,s,\ell-1}:=\Gamma_{\vb{o}_{<\ell}}^{m,s,\ell-1}-\Gamma_{\vb{o}_{<\ell}}^{m-1,s,\ell-1}.
\end{align}
For $m<\ell$, both hybrid sequences use $\Lambda_{\vb{e},s}$ in round $\ell$; therefore, their difference after that round is
\begin{align}
    \Omega_{m,\ell}^{\vb{e},s}-\Omega_{m-1,\ell}^{\vb{e},s}
    =
    \sum_{\vb{o}_{<\ell}}\mathcal{M}_{\vb{o}_{<\ell}}^{\ell}\!\left[
        \left(\id_{W_{\ell-1}}\otimes\Lambda_{\vb{e},s}\otimes\id_{A_{\ell}}\right)
        \left(\Delta_{\vb{o}_{<\ell}}^{m,s,\ell-1}\otimes\rho_{\ell}^{\vb{o}_{<\ell}}\right)
    \right].
\end{align}
The triangle inequality, trace-norm contractivity under the channel and the outcome-labeled map, and $\|X\otimes\rho\|_1=\|X\|_1$ for a density operator $\rho$ give
\begin{align}
    \left\|\Omega_{m,\ell}^{\vb{e},s}-\Omega_{m-1,\ell}^{\vb{e},s}\right\|_1
    &\leq
    \sum_{\vb{o}_{<\ell}}\left\|\Delta_{\vb{o}_{<\ell}}^{m,s,\ell-1}\right\|_1
    =
    \left\|\Omega_{m,\ell-1}^{\vb{e},s}-\Omega_{m-1,\ell-1}^{\vb{e},s}\right\|_1.
\end{align}
The last equality follows because branches with different previous measurement outcomes occupy orthogonal classical subspaces.

We next bound the trace-norm difference introduced at the $m$th channel use.
Before round $m$, the two hybrid sequences are identical; therefore, $\Gamma_{\vb{o}_{<m}}^{m,s,m-1}=\Gamma_{\vb{o}_{<m}}^{m-1,s,m-1}$ for every sequence of previous measurement outcomes.
In round $m$, the first sequence uses $\Lambda_{\dep}$ and the second uses $\Lambda_{\vb{e},s}$, while the fresh input state is the same for fixed $\vb{o}_{<m}$.
Using the outcome-labeled map and trace-norm contractivity gives
\begin{align}
    \left\|\Omega_{m,m}^{\vb{e},s}-\Omega_{m-1,m}^{\vb{e},s}\right\|_1
    &\leq
    \sum_{\vb{o}_{<m}}\Tr\left[\Gamma_{\vb{o}_{<m}}^{m,s,m-1}\right]
    \left\|\left[\left(\Lambda_{\dep}-\Lambda_{\vb{e},s}\right)\otimes\id_{A_m}\right]\!\left(\rho_m^{\vb{o}_{<m}}\right)\right\|_1.
\end{align}
The branch traces are the prefix probabilities $Q(\vb{o}_{<m})$ under the null hypothesis.
The preceding propagation bound and the hybrid estimate therefore imply
\begin{align}
    \TVD(Q,R_{\vb{e}})
    &\leq
    \frac{1}{4}\sum_{s=\pm1}\sum_{m=1}^{N}\sum_{\vb{o}_{<m}}
    Q(\vb{o}_{<m})
    \left\|\left[\left(\Lambda_{\dep}-\Lambda_{\vb{e},s}\right)\otimes\id_{A_m}\right]\!\left(\rho_m^{\vb{o}_{<m}}\right)\right\|_1.
\end{align}
For fixed $m$ and $\vb{o}_{<m}$, both $Q(\vb{o}_{<m})$ and $\rho_m^{\vb{o}_{<m}}$ are independent of $\vb{e}$ and $s$.
Averaging over nonzero $\vb{e}$ therefore gives
\begin{align}
    \mathds{E}_{\vb{e}\neq0}\TVD(Q,R_{\vb{e}})
    &\leq
    \frac{1}{4}\sum_{s=\pm1}\sum_{m=1}^{N}\sum_{\vb{o}_{<m}}
    Q(\vb{o}_{<m})
    \mathds{E}_{\vb{e}\neq0}\left\|\left[\left(\Lambda_{\dep}-\Lambda_{\vb{e},s}\right)\otimes\id_{A_m}\right]\!\left(\rho_m^{\vb{o}_{<m}}\right)\right\|_1.
\end{align}
Using the Pauli-error representation from Section~\ref{SM_Section_1}, the channel difference satisfies
\begin{align}
    \left[\left(\Lambda_{\dep}-\Lambda_{\vb{e},s}\right)\otimes\id_A\right](\rho)
    =
    -\frac{3s\epsilon}{4^n}\sum_{\vb{x}\in\mathbb{Z}_2^{2n}}(-1)^{\langle\vb{x},\vb{e}\rangle}\left(P_{\vb{x}}\otimes I_A\right)\rho\left(P_{\vb{x}}\otimes I_A\right).
\end{align}
The lemma below bounds the average trace norm of the normalized sum on the right-hand side.

\begin{lemma}\label{SM_lemma_quantum_adaptive_input_restriction}
Let $S$ be an $n$-qubit system and let $A$ be an arbitrary ancilla system.
For every state $\rho_{SA}$ satisfying $H_{\min}(S\mid A)_{\rho}\geq-\eta$, the following bound holds:
\begin{align}
    \mathds{E}_{\vb{e}\neq0}
    \left\|
        \frac{1}{4^n}
        \sum_{\vb{x}\in\mathbb{Z}_2^{2n}}
        (-1)^{\langle\vb{x},\vb{e}\rangle}
        \left(P_{\vb{x}}\otimes I_A\right)
        \rho_{SA}
        \left(P_{\vb{x}}\otimes I_A\right)
    \right\|_1
    \leq
    3\cdot2^{-(n-\eta)/2}.
\end{align}
\end{lemma}
Applying Lemma~\ref{SM_lemma_quantum_adaptive_input_restriction} to each fresh input state in the preceding averaged hybrid bound and using $\sum_{\vb{o}_{<m}}Q(\vb{o}_{<m})=1$ gives
\begin{align}
    \mathds{E}_{\vb{e}\neq0}\TVD(Q,R_{\vb{e}})
    \leq
    \frac{1}{4}
    \sum_{s=\pm1}
    \sum_{m=1}^{N}
    9\epsilon\,2^{-(n-\eta)/2}
    =
    \frac{9}{2}N\epsilon\,2^{-(n-\eta)/2}.
\end{align}
Combining this estimate with Eq.~\eqref{Eq_SM_Coherent_TVD_Lower} gives
\begin{align}
    N
    =
    \Omega\!\left(\epsilon^{-1}2^{(n-\eta)/2}\right).
\end{align}
This proves the input-restricted case of Theorem~2 and the corresponding coherent-adaptive case of Corollary~1.

\subsection{Measurement-restricted coherent protocols}
\subsubsection{Model description}
We retain the notation $\vec{\Lambda}:=(\Lambda_1,\ldots,\Lambda_N)$ for the ordered channel sequence.
For each round $\ell$, let $S_{\ell}$ denote the $n$-qubit system on which $\Lambda_{\ell}$ acts and let $A_{\ell}$ denote its associated ancilla.
The learner may prepare states that are entangled across different channel-use slots.
Conditioned on the previous measurement outcomes $\vb{o}_{<\ell}:=(o_1,\ldots,o_{\ell-1})$, the POVM $\{E_{o_{\ell}}^{\vb{o}_{<\ell}}\}_{o_{\ell}}$ acts only on the current pair $S_{\ell}A_{\ell}$.
No measurement acts jointly on $S_{\ell}A_{\ell}$ and the remaining channel-use registers.
For each fixed $\vb{o}_{<\ell}$ and every outcome with $\Tr E_{o_{\ell}}^{\vb{o}_{<\ell}}>0$, there exists a density operator $\sigma_{A_{\ell}}^{\vb{o}_{<\ell},o_{\ell}}$ such that
\begin{align}
    \frac{E_{o_{\ell}}^{\vb{o}_{<\ell}}}{\Tr E_{o_{\ell}}^{\vb{o}_{<\ell}}}
    &\leq
    2^\eta I_{S_{\ell}}\otimes\sigma_{A_{\ell}}^{\vb{o}_{<\ell},o_{\ell}},
    \qquad
    \sum_{o_{\ell}}\Tr\left(E_{o_{\ell}}^{\vb{o}_{<\ell}}\right)\sigma_{A_{\ell}}^{\vb{o}_{<\ell},o_{\ell}}
    =
    2^n I_{A_{\ell}}.
\end{align}
We denote these conditions by $H_{\min}(S_{\ell}\mid A_{\ell})_{\{E_{o_{\ell}}^{\vb{o}_{<\ell}}\}_{o_{\ell}}}\geq-\eta$.

\begin{figure}[t]
    \centering
    \includegraphics[width=0.97\textwidth]{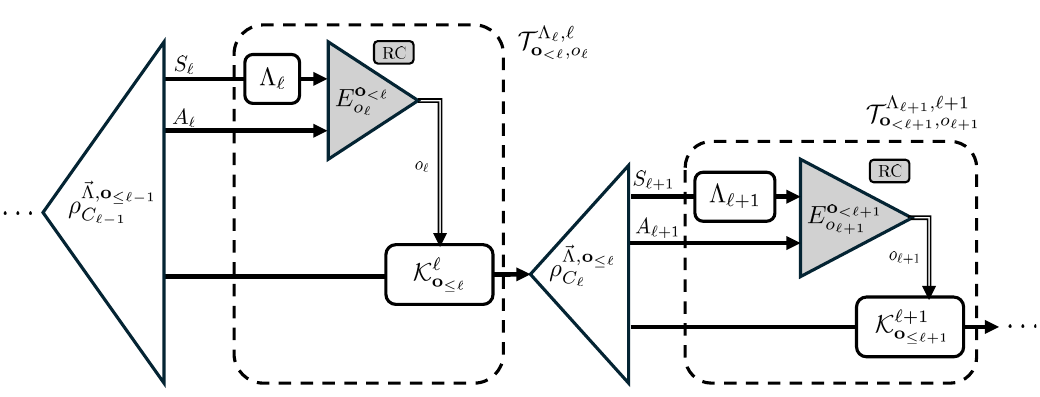}
\caption{
Measurement-restricted coherent adaptive protocol.
Before round $\ell$, the unconsumed channel-use registers may be entangled.
The channel $\Lambda_{\ell}$ acts on $S_{\ell}$, and the restricted POVM $\{E_{o_{\ell}}^{\vb{o}_{<\ell}}\}_{o_{\ell}}$ acts only on $S_{\ell}A_{\ell}$ and produces the measurement outcome $o_{\ell}$.
An outcome-dependent CPTP map $\mathcal{K}_{\vb{o}_{\leq\ell}}^{\ell}$ coherently processes the unmeasured registers and prepares them for subsequent rounds.
The dashed box represents the corresponding CPTNI branch map $\mathcal{T}_{\vb{o}_{<\ell},o_{\ell}}^{\Lambda_{\ell},\ell}$.
}
    \label{fig_SM_Quantum_Restricted_Measurement}
\end{figure}

For each round $\ell$, define the unmeasured register $C_{\ell}:=S_{\ell+1:N}A_{\ell+1:N}$, with $C_N$ one-dimensional.
Immediately before round $\ell$, let $\rho_{S_{\ell}A_{\ell}C_{\ell}}^{\vec{\Lambda},\vb{o}_{<\ell}} = \rho_{C_{\ell-1}}^{\vec{\Lambda},\vb{o}_{<\ell}}$ denote the subnormalized state of all unconsumed registers associated with the previous measurement outcomes $\vb{o}_{<\ell}$.
Its trace is the probability of these previous measurement outcomes:
\begin{align}
    \Tr\left[\rho_{S_{\ell}A_{\ell}C_{\ell}}^{\vec{\Lambda},\vb{o}_{<\ell}}\right]
    =
    \Pr_{\vec{\Lambda}}[\vb{o}_{<\ell}].
\end{align}
The dependence on $\vec{\Lambda}$ arises from the preceding channel uses and does not imply that the learner knows the channel sequence.
After $\Lambda_{\ell}$ acts and outcome $o_{\ell}$ is obtained, the subnormalized state of $C_{\ell}$ is
\begin{align}
    \widetilde{\rho}_{C_{\ell}}^{\vec{\Lambda},\vb{o}_{\leq\ell}}
    &:=
    \Tr_{S_{\ell}A_{\ell}}\left[
        \left(E_{o_{\ell}}^{\vb{o}_{<\ell}}\otimes I_{C_{\ell}}\right)
        \left(\Lambda_{\ell}\otimes\id_{A_{\ell}C_{\ell}}\right)
        \left(\rho_{S_{\ell}A_{\ell}C_{\ell}}^{\vec{\Lambda},\vb{o}_{<\ell}}\right)
    \right].
\end{align}
Conditioned on $\vb{o}_{\leq\ell}$, the learner applies a CPTP map $\mathcal{K}_{\vb{o}_{\leq\ell}}^{\ell}$ to $C_{\ell}$ to prepare the unconsumed registers for the subsequent rounds.
For fixed $\vb{o}_{<\ell}$ and $o_{\ell}$, define the branch map
\begin{align}
    \mathcal{T}_{\vb{o}_{<\ell},o_{\ell}}^{\Lambda_{\ell},\ell}(X)
    &:=
    \mathcal{K}_{\vb{o}_{\leq\ell}}^{\ell}\!\left(
        \Tr_{S_{\ell}A_{\ell}}\!\left[
            \left(E_{o_{\ell}}^{\vb{o}_{<\ell}}\otimes I_{C_{\ell}}\right)
            \left(\Lambda_{\ell}\otimes\id_{A_{\ell}C_{\ell}}\right)(X)
        \right]
    \right).
\end{align}
Each $\mathcal{T}_{\vb{o}_{<\ell},o_{\ell}}^{\Lambda_{\ell},\ell}$ is CPTNI, and $\sum_{o_{\ell}}\mathcal{T}_{\vb{o}_{<\ell},o_{\ell}}^{\Lambda_{\ell},\ell}$ is CPTP.
The corresponding outcome-labeled map is
\begin{align}
    \mathcal{M}_{\vb{o}_{<\ell}}^{\Lambda_{\ell},\ell}(X)
    &:=
    \sum_{o_{\ell}}
    \ketbra{\vb{o}_{\leq\ell}}{\vb{o}_{\leq\ell}}
    \otimes
    \mathcal{T}_{\vb{o}_{<\ell},o_{\ell}}^{\Lambda_{\ell},\ell}(X),
\end{align}
and it is CPTP.
The updated branch state is $\rho_{C_{\ell}}^{\vec{\Lambda},\vb{o}_{\leq\ell}}:=\mathcal{T}_{\vb{o}_{<\ell},o_{\ell}}^{\Lambda_{\ell},\ell}(\rho_{S_{\ell}A_{\ell}C_{\ell}}^{\vec{\Lambda},\vb{o}_{<\ell}}) = \mathcal{T}_{\vb{o}_{<\ell},o_{\ell}}^{\Lambda_{\ell},\ell}(\rho_{C_{\ell-1}}^{\vec{\Lambda},\vb{o}_{\leq \ell-1}})$, where $C_{\ell}$ is relabeled as $S_{\ell+1}A_{\ell+1}C_{\ell+1}$ before the next round.
At the $\ell$th causal cut, the classical--quantum state is
\begin{align}
    \Omega_{\ell}^{\vec{\Lambda}}
    &:=
    \sum_{\vb{o}_{\leq\ell}}
    \ketbra{\vb{o}_{\leq\ell}}{\vb{o}_{\leq\ell}}
    \otimes
    \rho_{C_{\ell}}^{\vec{\Lambda},\vb{o}_{\leq\ell}}.
\end{align}

\subsubsection{Lower-bound proof}
For $m=0,\ldots,N$ and $\ell=1,\ldots,N$, define $\Omega_{m,\ell}^{\vb{e},s}:=\Omega_{\ell}^{\vec{\Lambda}_{\vb{e},m}^{s}}$.
For fixed nonzero $\vb{e}$ and $s\in\{\pm1\}$, let $R_{\vb{e},s}(\vb{o}_{\leq N}):=\Pr_{\vec{\Lambda}_{\vb{e},0}^{s}}[\vb{o}_{\leq N}]$ and $Q(\vb{o}_{\leq N}):=\Pr_{\vec{\Lambda}_{\vb{e},N}^{s}}[\vb{o}_{\leq N}]$.
The sign-averaged alternative distribution is $R_{\vb{e}}(\vb{o}_{\leq N})=\mathds{E}_{s}R_{\vb{e},s}(\vb{o}_{\leq N})$.
Because $C_N$ is one-dimensional, the trace norm of the final classical--quantum state difference equals the $\ell_1$ distance between the corresponding measurement-outcome distributions.
For each fixed nonzero $\vb{e}$, the triangle inequality for the sign average and the hybrid telescoping sum give
\begin{align}
    \TVD(Q,R_{\vb{e}})
    &\leq
    \frac{1}{4}
    \sum_{s=\pm1}
    \left\|
        \Omega_{N,N}^{\vb{e},s}
        -
        \Omega_{0,N}^{\vb{e},s}
    \right\|_1
    \leq
    \frac{1}{4}
    \sum_{s=\pm1}
    \sum_{m=1}^{N}
    \left\|
        \Omega_{m,N}^{\vb{e},s}
        -
        \Omega_{m-1,N}^{\vb{e},s}
    \right\|_1.
\end{align}
We first show that rounds after the differing channel use cannot increase the trace-norm difference.
Fix $m<\ell$ and define
\begin{align}
    \Delta_{\vb{o}_{<\ell}}^{m,s,\ell-1}
    &:=
    \rho_{S_{\ell}A_{\ell}C_{\ell}}^{\vec{\Lambda}_{\vb{e},m}^{s},\vb{o}_{<\ell}}
    -
    \rho_{S_{\ell}A_{\ell}C_{\ell}}^{\vec{\Lambda}_{\vb{e},m-1}^{s},\vb{o}_{<\ell}}.
\end{align}
For $m<\ell$, both hybrid sequences use $\Lambda_{\vb{e},s}$ in round $\ell$, and the protocol operations depend only on the previous measurement outcomes.
Therefore, the same outcome-labeled CPTP map acts on both branch states, and
\begin{align}
    \Omega_{m,\ell}^{\vb{e},s}
    -
    \Omega_{m-1,\ell}^{\vb{e},s}
    =
    \sum_{\vb{o}_{<\ell}}
    \mathcal{M}_{\vb{o}_{<\ell}}^{\Lambda_{\vb{e},s},\ell}\!\left(
        \Delta_{\vb{o}_{<\ell}}^{m,s,\ell-1}
    \right).
\end{align}
The triangle inequality and trace-norm contractivity give
\begin{align}
    \left\|
        \Omega_{m,\ell}^{\vb{e},s}
        -
        \Omega_{m-1,\ell}^{\vb{e},s}
    \right\|_1
    &\leq
    \sum_{\vb{o}_{<\ell}}
    \left\|
        \Delta_{\vb{o}_{<\ell}}^{m,s,\ell-1}
    \right\|_1
    =
    \left\|
        \Omega_{m,\ell-1}^{\vb{e},s}
        -
        \Omega_{m-1,\ell-1}^{\vb{e},s}
    \right\|_1.
\end{align}
The final equality follows because distinct sequences of previous measurement outcomes occupy orthogonal classical subspaces.

We next bound the trace-norm difference introduced in round $m$.
Before round $m$, the two hybrid sequences are identical; therefore, for every sequence of previous measurement outcomes,
\begin{align}
    \rho_{S_mA_mC_m}^{\vec{\Lambda}_{\vb{e},m}^{s},\vb{o}_{<m}}
    =
    \rho_{S_mA_mC_m}^{\vec{\Lambda}_{\vb{e},m-1}^{s},\vb{o}_{<m}}.
\end{align}
The common branch trace is the null prefix probability $Q(\vb{o}_{<m})$.
For every $\vb{o}_{<m}$ satisfying $Q(\vb{o}_{<m})>0$, define the normalized conditional state
\begin{align}
    \tau_m^{\vb{o}_{<m}}
    :=
    \frac{
        \rho_{S_mA_mC_m}^{\vec{\Lambda}_{\vb{e},m}^{s},\vb{o}_{<m}}
    }{
        Q(\vb{o}_{<m})
    }.
\end{align}
Sequences satisfying $Q(\vb{o}_{<m})=0$ make no contribution and are omitted below.
Because the first $m-1$ channels are depolarizing in both hybrid sequences, $Q(\vb{o}_{<m})$ and $\tau_m^{\vb{o}_{<m}}$ are independent of $\vb{e}$ and $s$.
Trace-norm contractivity under the outcome-dependent CPTP maps $\mathcal{K}_{\vb{o}_{\leq m}}^{m}$ gives
\begin{align}
    \left\|\Omega_{m,m}^{\vb{e},s}-\Omega_{m-1,m}^{\vb{e},s}\right\|_1
    &\leq
    \sum_{\vb{o}_{<m}}Q(\vb{o}_{<m})
    \sum_{o_m}
    \left\|
        \Tr_{S_mA_m}\left[
            \left(E_{o_m}^{\vb{o}_{<m}}\otimes I_{C_m}\right)
            \left[
                \left(\Lambda_{\dep}-\Lambda_{\vb{e},s}\right)
                \otimes\id_{A_mC_m}
            \right]
            \!\left(\tau_m^{\vb{o}_{<m}}\right)
        \right]
    \right\|_1.
\end{align}
Using the Pauli-error representation from Section~\ref{SM_Section_1}, the channel difference satisfies
\begin{align}
    &
    \left[
        \left(\Lambda_{\dep}-\Lambda_{\vb{e},s}\right)
        \otimes\id_{A_mC_m}
    \right](\tau)
    =
    -\frac{3s\epsilon}{4^n}
    \sum_{\vb{x}\in\mathbb{Z}_2^{2n}}
    (-1)^{\langle\vb{x},\vb{e}\rangle}
    \left(P_{\vb{x}}\otimes I_{A_mC_m}\right)
    \tau
    \left(P_{\vb{x}}\otimes I_{A_mC_m}\right).
\end{align}
The following lemma bounds the average of the resulting expression.
\begin{lemma}\label{SM_lemma_quantum_adaptive_measurement_restriction}
Let $S$ be an $n$-qubit system, and let $A$ and $C$ be arbitrary quantum systems.
Let $\rho_{SAC}$ be a density operator, and let $\{E_o\}_o$ be a POVM on $SA$.
Suppose that for every outcome with $\Tr E_o>0$ there exists a density operator $\sigma_A^o$ such that
\begin{align}
    \frac{E_o}{\Tr E_o}
    \leq
    2^\eta I_S\otimes\sigma_A^o,
    \qquad
    \sum_o\Tr(E_o)\sigma_A^o
    =
    2^n I_A.
\end{align}
Then
\begin{align}
    &
    \mathds{E}_{\vb{e}\neq0}
    \sum_o
    \left\|
        \frac{1}{4^n}
        \sum_{\vb{x}\in\mathbb{Z}_2^{2n}}
        (-1)^{\langle\vb{x},\vb{e}\rangle}
        \Tr_{SA}\left[
            \left(E_o\otimes I_C\right)
            \left(P_{\vb{x}}\otimes I_{AC}\right)
            \rho_{SAC}
            \left(P_{\vb{x}}\otimes I_{AC}\right)
        \right]
    \right\|_1
    \leq
    3\cdot2^{-(n-\eta)/2}.
\end{align}
\end{lemma}
Applying Lemma~\ref{SM_lemma_quantum_adaptive_measurement_restriction} to each $\tau_m^{\vb{o}_{<m}}$ and using $\sum_{\vb{o}_{<m}}Q(\vb{o}_{<m})=1$ gives
\begin{align}
    \mathds{E}_{\vb{e}\neq0}
    \left\|\Omega_{m,m}^{\vb{e},s}-\Omega_{m-1,m}^{\vb{e},s}\right\|_1
    \leq
    9\epsilon\,2^{-(n-\eta)/2}.
\end{align}
Combining this estimate with the propagation bound and the hybrid telescoping sum gives
\begin{align}
    \mathds{E}_{\vb{e}\neq0}\TVD(Q,R_{\vb{e}})
    &\leq
    \frac{1}{4}
    \sum_{s=\pm1}
    \sum_{m=1}^{N}
    9\epsilon\,2^{-(n-\eta)/2}
    =
    \frac{9}{2}N\epsilon\,2^{-(n-\eta)/2}.
\end{align}
Combining this upper bound with Eq.~\eqref{Eq_SM_Coherent_TVD_Lower} yields
\begin{align}
    N
    =
    \Omega\!\left(\epsilon^{-1}2^{(n-\eta)/2}\right).
\end{align}
This proves the measurement-restricted case of Theorem~2 and the corresponding coherent-adaptive case of Corollary~1.

\subsection{Auxiliary lemmas}
The following lemmas establish the single-round trace-norm bounds used in the two coherent-adaptive lower-bound proofs.
We first derive a one-bit Fourier consequence of the leftover-hash lemma and then apply it to the input- and measurement-restricted models.
\subsubsection{A one-bit application of the leftover-hash lemma}

\begin{lemma}
\label{SM_lemma_LHL_2}
Let
\begin{align}
    \omega_{XE}
    :=
    \sum_{\vb{x}\in\mathbb{Z}_2^{2n}}
    \ketbra{\vb{x}}{\vb{x}}_X
    \otimes
    \omega_E^{\vb{x}}
\end{align}
be a normalized classical--quantum state, where each $\omega_E^{\vb{x}}$ is a subnormalized positive semidefinite operator.
Suppose that
\begin{align}
    H_{\min}(X\mid E)_\omega
    \geq
    h.
\end{align}
Then
\begin{align}
    \mathds{E}_{\vb{u}\in\mathbb{Z}_2^{2n}}
    \left\|
        \sum_{\vb{x}\in\mathbb{Z}_2^{2n}}
        (-1)^{\langle\vb{x},\vb{u}\rangle}
        \omega_E^{\vb{x}}
    \right\|_1
    \leq
    2^{(1-h)/2}.
    \label{Eq_SM_OneBitFourierLHL}
\end{align}
Here, $\langle\cdot,\cdot\rangle$ is the symplectic inner product defined in Section~\ref{SM_Section_1}.
\end{lemma}

\begin{proof}
For each $\vb{u}\in\mathbb{Z}_2^{2n}$, define $f_{\vb{u}}(\vb{x}):=\langle\vb{x},\vb{u}\rangle$.
For distinct $\vb{x},\vb{x}^{\prime}\in\mathbb{Z}_2^{2n}$,
\begin{align}
    \Pr_{\vb{u}}\!\left[
        f_{\vb{u}}(\vb{x})
        =
        f_{\vb{u}}(\vb{x}^{\prime})
    \right]
    =
    \Pr_{\vb{u}}\!\left[
        \langle\vb{x}\oplus\vb{x}^{\prime},\vb{u}\rangle
        =
        0
    \right]
    =
    \frac{1}{2}.
\end{align}
Therefore, $\mathcal{F}:=\{f_{\vb{u}}:\vb{u}\in\mathbb{Z}_2^{2n}\}$ is two-universal.
For $\vb{u}\in\mathbb{Z}_2^{2n}$ and $z\in\mathbb{Z}_2$, define
\begin{align}
    \omega_E^{[f_{\vb{u}},z]}
    &:=
    \sum_{\substack{
        \vb{x}\in\mathbb{Z}_2^{2n}\\
        f_{\vb{u}}(\vb{x})=z
    }}
    \omega_E^{\vb{x}},
    \qquad
    \omega_{ZFE}
    :=
    \frac{1}{|\mathcal{F}|}
    \sum_{z\in\mathbb{Z}_2}
    \sum_{f_{\vb{u}}\in\mathcal{F}}
    \ketbra{z}{z}_Z
    \otimes
    \ketbra{f_{\vb{u}}}{f_{\vb{u}}}_F
    \otimes
    \omega_E^{[f_{\vb{u}},z]}.
\end{align}
The leftover-hash lemma~\cite{Tomamichel_2011}, applied to the two-universal family $\mathcal{F}$ with one-bit output, gives
\begin{align}
    \min_{\sigma_{FE}}
    \frac{1}{2}
    \left\|
        \omega_{ZFE}
        -
        \frac{1}{2}I_Z\otimes\sigma_{FE}
    \right\|_1
    \leq
    \frac{1}{2}\sqrt{2^{1-h}},
\end{align}
where the minimum is over density operators $\sigma_{FE}$.
For $z\in\mathbb{Z}_2$, define
\begin{align}
    \omega_{FE}^{(z)}
    :=
    \frac{1}{|\mathcal{F}|}
    \sum_{f_{\vb{u}}\in\mathcal{F}}
    \ketbra{f_{\vb{u}}}{f_{\vb{u}}}_F
    \otimes
    \omega_E^{[f_{\vb{u}},z]}.
\end{align}
Then $\omega_{ZFE}=\sum_{z\in\mathbb{Z}_2}\ketbra{z}{z}_Z\otimes\omega_{FE}^{(z)}$.
Let $\sigma_{FE}$ attain the preceding minimum.
Block diagonality in $Z$ and the triangle inequality give
\begin{align}
    \sqrt{2^{1-h}}
    &\geq
    \left\|
        \omega_{ZFE}
        -
        \frac{1}{2}I_Z\otimes\sigma_{FE}
    \right\|_1
    \geq
    \left\|
        \omega_{FE}^{(0)}
        -
        \omega_{FE}^{(1)}
    \right\|_1
    =
    \mathds{E}_{\vb{u}}
    \left\|
        \omega_E^{[f_{\vb{u}},0]}
        -
        \omega_E^{[f_{\vb{u}},1]}
    \right\|_1
    =
    \mathds{E}_{\vb{u}}
    \left\|
        \sum_{\vb{x}\in\mathbb{Z}_2^{2n}}
        (-1)^{\langle\vb{x},\vb{u}\rangle}
        \omega_E^{\vb{x}}
    \right\|_1.
\end{align}
This proves Lemma~\ref{SM_lemma_LHL_2}.
\end{proof}

\subsubsection{Proof of Lemma~\ref{SM_lemma_quantum_adaptive_input_restriction}}

\begin{proof}
By the assumption $H_{\min}(S\mid A)_{\rho}\geq-\eta$, there exists a density operator $\sigma_A$ such that
\begin{align}
    \rho_{SA}
    \leq
    2^\eta
    \left(
        I_S\otimes\sigma_A
    \right).
\end{align}
We encode the uniformly Pauli-conjugated states into a classical--quantum state and apply Lemma~\ref{SM_lemma_LHL_2}.
For each $\vb{x}\in\mathbb{Z}_2^{2n}$, define
\begin{align}
    b_{\vb{x}}
    :=
    \left(
        P_{\vb{x}}\otimes I_A
    \right)
    \rho_{SA}
    \left(
        P_{\vb{x}}\otimes I_A
    \right)
    \leq
    2^\eta
    \left(
        I_S\otimes\sigma_A
    \right).
\end{align}
Define the normalized classical--quantum state
\begin{align}
    \omega_{XSA}
    :=
    \frac{1}{4^n}
    \sum_{\vb{x}\in\mathbb{Z}_2^{2n}}
    \ketbra{\vb{x}}{\vb{x}}_X
    \otimes
    b_{\vb{x}}.
\end{align}
The bound on $b_{\vb{x}}$ implies
\begin{align}
    \omega_{XSA}
    \leq
    2^{-(n-\eta)}
    I_X
    \otimes
    \left(
        2^{-n}I_S\otimes\sigma_A
    \right).
\end{align}
Since $2^{-n}I_S\otimes\sigma_A$ is a density operator, it follows that
\begin{align}
    H_{\min}(X\mid SA)_\omega
    \geq
    n-\eta.
\end{align}
Applying Lemma~\ref{SM_lemma_LHL_2} and then conditioning the uniform average on $\vb{e}\neq0$ gives
\begin{align}
    &
    \mathds{E}_{\vb{e}\neq0}
    \left\|
        \frac{1}{4^n}
        \sum_{\vb{x}\in\mathbb{Z}_2^{2n}}
        (-1)^{\langle\vb{x},\vb{e}\rangle}
        \left(
            P_{\vb{x}}\otimes I_A
        \right)
        \rho_{SA}
        \left(
            P_{\vb{x}}\otimes I_A
        \right)
    \right\|_1
    \leq
    \frac{4^n}{4^n-1}
    2^{(1-n+\eta)/2}
    \leq
    3 \cdot 2^{-(n-\eta)/2}.
\end{align}
This proves Lemma~\ref{SM_lemma_quantum_adaptive_input_restriction}.
\end{proof}

\subsubsection{Proof of Lemma~\ref{SM_lemma_quantum_adaptive_measurement_restriction}}

\begin{proof}
The proof follows the same strategy as the input-restricted case, but we additionally retain the measurement outcome in a classical register.
For each $\vb{x}\in\mathbb{Z}_2^{2n}$ and each POVM outcome $o$, define
\begin{align}
    X_{\vb{x}}
    &:=
    \left(P_{\vb{x}}\otimes I_{AC}\right)
    \rho_{SAC}
    \left(P_{\vb{x}}\otimes I_{AC}\right),
    \qquad
    B_o(\vb{x})
    :=
    \Tr_{SA}\left[
        \left(E_o\otimes I_C\right)
        X_{\vb{x}}
    \right].
\end{align}
For every outcome with $\Tr E_o>0$, define
\begin{align}
    F_o
    &:=
    2^\eta\Tr(E_o)
    \left(I_S\otimes\sigma_A^o\right),
    \qquad
    D_o
    :=
    2^\eta\Tr(E_o)
    \Tr_{SA}\left[
        \left(I_S\otimes\sigma_A^o\otimes I_C\right)
        \rho_{SAC}
    \right].
\end{align}
For an outcome with $\Tr E_o=0$, we have $E_o=0$ and set $F_o=D_o=0$.
The measurement restriction implies $F_o\geq E_o$.
Because $F_o$ acts as the identity on $S$, it commutes with $P_{\vb{x}}\otimes I_A$, and hence $\Tr_{SA}[(F_o\otimes I_C)X_{\vb{x}}]=D_o$.
Using cyclicity of the partial trace over $SA$, we obtain
\begin{align}
    D_o-B_o(\vb{x})
    &=
    \Tr_{SA}\left[
        \left(\sqrt{F_o-E_o}\otimes I_C\right)
        X_{\vb{x}}
        \left(\sqrt{F_o-E_o}\otimes I_C\right)
    \right]
    \geq
    0.
\end{align}
Thus, $B_o(\vb{x})\leq D_o$ for every $\vb{x}$ and $o$.
The second condition in the measurement restriction gives
\begin{align}
    \sum_o\Tr D_o
    =
    2^\eta
    \Tr\left[
        \left(
            I_S\otimes
            \sum_o\Tr(E_o)\sigma_A^o
            \otimes I_C
        \right)
        \rho_{SAC}
    \right]
    =
    2^{n+\eta}.
\end{align}
Define
\begin{align}
    \omega_{XOC}
    &:=
    \frac{1}{4^n}
    \sum_{\vb{x}\in\mathbb{Z}_2^{2n}}
    \ketbra{\vb{x}}{\vb{x}}_X
    \otimes
    \sum_o
    \ketbra{o}{o}_O
    \otimes
    B_o(\vb{x}),
    \qquad
    \sigma_{OC}
    :=
    2^{-(n+\eta)}
    \sum_o
    \ketbra{o}{o}_O
    \otimes
    D_o.
\end{align}
The completeness of the POVM implies $\sum_o\Tr B_o(\vb{x})=1$ for every $\vb{x}$; therefore, $\omega_{XOC}$ is normalized.
The preceding trace identity shows that $\sigma_{OC}$ is a density operator.
The bound $B_o(\vb{x})\leq D_o$ implies
\begin{align}
    \omega_{XOC}
    \leq
    2^{-(n-\eta)}
    I_X\otimes\sigma_{OC}.
\end{align}
Since $\sigma_{OC}$ is a density operator, it follows that
\begin{align}
    H_{\min}(X\mid OC)_\omega
    \geq
    n-\eta.
\end{align}

Applying Lemma~\ref{SM_lemma_LHL_2}, using block diagonality in $O$, and then conditioning the uniform average on $\vb{e}\neq0$ gives
\begin{align}
    &
    \mathds{E}_{\vb{e}\neq0}
    \sum_o
    \left\|
        \frac{1}{4^n}
        \sum_{\vb{x}\in\mathbb{Z}_2^{2n}}
        (-1)^{\langle\vb{x},\vb{e}\rangle}
        B_o(\vb{x})
    \right\|_1\leq
    \frac{4^n}{4^n-1}
    2^{(1-n+\eta)/2}
    \leq
    3 \cdot 2^{-(n-\eta)/2}.
\end{align}
This proves Lemma~\ref{SM_lemma_quantum_adaptive_measurement_restriction}.
\end{proof}

\section{Bounds for conjugate-state learning}

In this section, we prove Theorem~3 and Corollary~2 of the main text.
We first recall the precise statements.
\medskip

\noindent\textbf{Theorem 3 (main text).}
Assume $d\geq2$ and $0<\epsilon<1/6$.
Consider a copywise adaptive protocol in which each copy of
$\rho\otimes\rho^*$ is measured separately and no quantum memory is retained across distinct copies.
If every normalized nonzero POVM effect satisfies the reduction-criterion restriction, then learning all displacement amplitudes with success probability at least $2/3$ requires
\begin{align}
    N=\Omega\left(\epsilon^{-2}\sqrt{d}\right)
\end{align}
copies.

\medskip
\noindent\textbf{Corollary 2 (main text).}
Under the same protocol model, for some fixed $\eta\geq0$, suppose that every normalized nonzero POVM effect satisfies
\begin{align}
    \frac{M_o}{\Tr M_o}\leq2^\eta I_S\otimes\sigma_A^o
\end{align}
for some density operator $\sigma_A^o$.
Then learning all displacement amplitudes with success probability at least $2/3$ requires
\begin{align}
    N=\Omega\left(\epsilon^{-2}\sqrt{d\,2^{-\eta}}\right).
    \label{Eq_SM_Conjugate_Target}
\end{align}

We prove the quantitative bound in Eq.~\eqref{Eq_SM_Conjugate_Target}.
Setting $\eta=0$ yields Theorem~3 because every normalized effect satisfying the reduction criterion obeys the condition above.

The proof follows the hypothesis-testing strategy used for Pauli-channel learning.
We construct states whose displacement amplitude is perturbed in a uniformly random direction and reveal this direction only after all measurements have been completed.
For the conjugate-state pair $\rho\otimes\rho^*$, the perturbation contains both a term linear in its random sign and a sign-independent quadratic term.
After sign averaging, Lemma~\ref{SM_lemma_StateLearning_0} reduces the total variation distance to two averaged response quantities.
Lemmas~\ref{SM_lemma_StateLearning_1} and~\ref{SM_lemma_StateLearning_2} bound these quantities by $O(d^{-1})$ and
$O(\sqrt{2^\eta/d})$, respectively, yielding Eq.~\eqref{Eq_SM_Conjugate_Target}.

\subsection{Preliminaries}

We collect the displacement-operator identities used in the hard-instance construction and in the proof of Lemma~\ref{SM_lemma_StateLearning_2}.
Let
$\mathcal I_d:=\{0,1,\ldots,d-1\}$ and
$\omega:=e^{2\pi i/d}$.
Elements of $\mathbb Z_d$ are identified with their canonical representatives in $\mathcal I_d$, while signed integers are retained when they appear in the phase of a displacement operator.

Let $X$ and $Z$ denote the shift and clock operators, and follow the convention of Ref.~\cite{King_2024}:
\begin{align}
X\ket{j}
&=
\ket{j+1},
\qquad
Z\ket{j}
=
\omega^j\ket{j},
\qquad
ZX
=
\omega XZ,
\qquad
D_{q,p}
:=
e^{i\pi qp/d}X^qZ^p.
\end{align}
For a state $\rho$, its displacement amplitudes are
$y_{q,p}(\rho):=\Tr(D_{q,p}\rho)$.
The displacement operators satisfy
\begin{align}
D_{q,p}\ket{j}
&=
e^{i\pi(q+2j)p/d}\ket{j+q},
\quad
D_{q,p}^{\dagger}
=
D_{-q,-p},
\quad
D_{q,p}^{*}
=
D_{q,-p},
\quad
D_{q,p}^{T}
=
D_{-q,p},
\quad
\Tr\left(
D_{q,p}^{\dagger}D_{q',p'}
\right)
=
d\delta_{q,q'}\delta_{p,p'}.
\end{align}

We use the following two maximally entangled bases:
\begin{align}
\ket{\Phi_{a,b}}
&:=
\frac{1}{\sqrt d}
\sum_{j\in\mathcal I_d}
\omega^{bj}
\ket{j+a}\ket{-j},
\qquad
\ket{\Psi_{a,b}}
:=
\frac{1}{\sqrt d}
\sum_{j\in\mathcal I_d}
\omega^{bj}
\ket{j+a}\ket{j}.
\end{align}
They obey
\begin{align}
\braket{\Phi_{a,b}}{\Phi_{a',b'}}
=
\delta_{a,a'}\delta_{b,b'}, &\quad
\braket{\Psi_{a,b}}{\Psi_{a',b'}}
=
\delta_{a,a'}\delta_{b,b'}, \\
\left(
D_{q,p}\otimes D_{q,p}^{T}
\right)
\ket{\Phi_{a,b}}
=
\omega^{ap-bq}\ket{\Phi_{a,b}}, & \quad
\left(
D_{q,p}\otimes D_{q,p}^{*}
\right)
\ket{\Psi_{a,b}}
=
\omega^{ap-bq}\ket{\Psi_{a,b}}.
\end{align}
Consequently,
\begin{align}
D_{q,p}\otimes D_{q,p}^{T}
=
\sum_{a,b\in\mathcal I_d}
\omega^{ap-bq}
\ketbra{\Phi_{a,b}}{\Phi_{a,b}},
\qquad
D_{q,p}\otimes D_{q,p}^{*}
=
\sum_{a,b\in\mathcal I_d}
\omega^{ap-bq}
\ketbra{\Psi_{a,b}}{\Psi_{a,b}}.
\end{align}

\subsection{Lower-bound proof}
Assume $0<\epsilon<1/6$ and define $\epsilon_1:=3\sqrt{2}\epsilon$.
For each nonzero $(q,p)\in\mathbb{Z}_d^2$ and $r\in\{\pm1\}$, consider the hypotheses
\begin{align}
    H_0:\quad
    \rho_0
    &=
    \frac{I}{d},
    \qquad
    H_1:\quad
    \rho_{q,p,r}
    =
    \frac{1}{d}
    \left(
        I+r\epsilon_1E_{q,p}
    \right).
    \label{Eq_SM_Hypothesis_state_learning}
\end{align}
Here,
\begin{align}
    E_{q,p}
    &:=
    \chi D_{q,p}
    +
    \chi^*D_{-q,-p},
    \qquad
    \chi
    :=
    \frac{1+i}{2}.
\end{align}
These operators satisfy~\cite{King_2024,asadian2016heisenberg}
\begin{align}
    E_{q,p}^{\dagger}
    =
    E_{q,p},
    \qquad
    E_{q,p}^{*}
    =
    E_{q,p}^{T}
    =
    E_{-q,p},
    \qquad
    \|E_{q,p}\|_{\operatorname{op}}
    \leq
    \sqrt{2}, \qquad
    \Tr\left(
        E_{q,p}E_{q',p'}
    \right)
    =
    d\,
    \delta_{q,q'}
    \delta_{p,p'}.
\end{align}
For nonzero $(q,p)$, the operator $E_{q,p}$ is traceless, and $\epsilon_1\|E_{q,p}\|_{\operatorname{op}}\leq6\epsilon<1$.
Hence, every $\rho_{q,p,r}$ is a density operator.
Moreover,
\begin{align}
    \left|
        \Tr\left(
            D_{q,p}\rho_{q,p,r}
        \right)
    \right|
    \geq
    3\epsilon,
    \qquad
    \Tr\left(
        D_{q,p}\rho_0
    \right)
    =
    0.
\end{align}
Here, the first inequality follows from
$\left|\Tr(D_{q,p}E_{q,p})\right|\geq d/\sqrt{2}$ and
$\epsilon_1=3\sqrt{2}\epsilon$.
Let $\hat{y}_{q,p}$ denote the learner's estimate of $y_{q,p}(\rho)$.
For the $m$th copy of $\rho\otimes\rho^*$, let $S_m$ and $A_m$ denote the registers containing $\rho$ and $\rho^*$, respectively.
Consider the following hypothesis-testing game.
\begin{enumerate}
    \item The referee samples a nonzero displacement index $z=(q,p)\in\mathbb{Z}_d^2\setminus\{(0,0)\}$ and a sign $r\in\{\pm1\}$ uniformly at random.
    \item The referee selects $H_0$ or $H_1$ from Eq.~\eqref{Eq_SM_Hypothesis_state_learning} with equal probability and gives the player $N$ copies of the corresponding conjugate-state pair.
    \item In round $m$, the player measures only $S_mA_m$ using a POVM $\{M_{o_m}^{\vb{o}_{<m}}\}_{o_m}$ selected according to the previous measurement outcomes.
    \item Every normalized nonzero POVM effect satisfies $M_{o_m}^{\vb{o}_{<m}}/\Tr M_{o_m}^{\vb{o}_{<m}}\leq2^\eta I_{S_m}\otimes\sigma_{o_m}^{\vb{o}_{<m}}$ for some density operator $\sigma_{o_m}^{\vb{o}_{<m}}$ on $A_m$, and all copies are measured before the referee reveals $z$.
    \item The referee reveals $z$, but not $r$, and asks the player to identify the selected hypothesis.
\end{enumerate}
The proof below uses only the effect-wise inequality in item~4; hence, it applies in particular to the POVM class appearing in Corollary~2.
After $z=(q,p)$ is revealed, the player declares $H_1$ if $|\hat{y}_{q,p}|>3\epsilon/2$ and declares $H_0$ otherwise.
On the learner's success event, $H_0$ implies $|\hat{y}_{q,p}|\leq\epsilon$, whereas $H_1$ and the reverse triangle inequality imply $|\hat{y}_{q,p}|\geq2\epsilon$.
Thus, a learner with success probability at least $2/3$ yields a testing strategy with success probability at least $2/3$.
Consequently, any sample-complexity lower bound for this testing game also applies to copywise adaptive conjugate-state learning.

Let $Q(\vb{o}_{1:N})$ denote the measurement-outcome distribution under $H_0$, and let $R_{z,r}(\vb{o}_{1:N})$ denote the corresponding distribution under $H_1$ with spike location $z=(q,p)$ and sign $r$.
Because the sign $r$ is not revealed, define the sign-averaged alternative distribution
\begin{align}
    R_z(\vb{o}_{1:N})
    :=
    \mathds{E}_{r=\pm1}R_{z,r}(\vb{o}_{1:N}).
\end{align}
Le Cam's two-point method~\cite{lecam1973convergence} and the preceding testing reduction imply
\begin{align}
    \frac{1}{3}
    \leq
    \mathds{E}_{z\neq(0,0)}
    \TVD\!\left(
        Q(\vb{o}_{1:N}),
        R_z(\vb{o}_{1:N})
    \right).
    \label{Eq_SM_StateLearning_TVD_Lower}
\end{align}
Writing $\overline{z}:=(-q,p)$, we have
\begin{align}
    \rho_{q,p,r}\otimes\rho_{q,p,r}^{*}
    =
    \frac{1}{d^2}
    \left[
        I\otimes I
        +
        r\epsilon_1
        \left(
            E_z\otimes I
            +
            I\otimes E_{\overline{z}}
        \right)
        +
        \epsilon_1^2
        E_z\otimes E_{\overline{z}}
    \right].
\end{align}
For every outcome with $\Tr M_{o_m}^{\vb{o}_{<m}}>0$, define the normalized POVM effect
\begin{align}
    \tau_{o_m}^{\vb{o}_{<m}}
    :=
    \frac{M_{o_m}^{\vb{o}_{<m}}}{\Tr M_{o_m}^{\vb{o}_{<m}}}.
\end{align}
For such an outcome, define
\begin{align}
    G_z^{\vb{o}_{<m}}(o_m)
    &:=
    \Tr\left[
        \left(
            E_z\otimes I
            +
            I\otimes E_{\overline{z}}
        \right)
        \tau_{o_m}^{\vb{o}_{<m}}
    \right],
    \qquad
    H_z^{\vb{o}_{<m}}(o_m)
    :=
    \Tr\left[
        \left(
            E_z\otimes E_{\overline{z}}
        \right)
        \tau_{o_m}^{\vb{o}_{<m}}
    \right].
\end{align}
The conditional measurement-outcome probabilities satisfy
\begin{align}
    Q(o_m\mid\vb{o}_{<m})
    &=
    \frac{\Tr M_{o_m}^{\vb{o}_{<m}}}{d^2},
    \qquad
    R_{z,r}(o_m\mid\vb{o}_{<m})
    =
    Q(o_m\mid\vb{o}_{<m})
    \left[
        1
        +
        r\epsilon_1G_z^{\vb{o}_{<m}}(o_m)
        +
        \epsilon_1^2H_z^{\vb{o}_{<m}}(o_m)
    \right].
\end{align}
If $\Tr M_{o_m}^{\vb{o}_{<m}}=0$, then $M_{o_m}^{\vb{o}_{<m}}=0$, and we set $G_z^{\vb{o}_{<m}}(o_m)=H_z^{\vb{o}_{<m}}(o_m)=0$.
For brevity, write $\vb{o}:=\vb{o}_{1:N}$ and $[x]_+:=\max\{x,0\}$.
Measurement-outcome sequences satisfying $Q(\vb{o})=0$ do not contribute to the TVD.
For every sequence satisfying $Q(\vb{o})>0$, the chain rule gives
\begin{align}
    \left[
        Q(\vb{o})-R_z(\vb{o})
    \right]_+
    =
    Q(\vb{o})
    \left[
        1
        -
        \mathds{E}_{r=\pm1}
        \prod_{m=1}^{N}
        \left(
            1
            +
            r\epsilon_1G_z^{\vb{o}_{<m}}(o_m)
            +
            \epsilon_1^2H_z^{\vb{o}_{<m}}(o_m)
        \right)
    \right]_+.
\end{align}
For every branch with $Q(o_m\mid\vb{o}_{<m})>0$, we have
\begin{align}
    1
    +
    r\epsilon_1G_z^{\vb{o}_{<m}}(o_m)
    +
    \epsilon_1^2H_z^{\vb{o}_{<m}}(o_m)
    =
    \frac{R_{z,r}(o_m\mid\vb{o}_{<m})}{Q(o_m\mid\vb{o}_{<m})}
    \geq
    0,
    \qquad
    r\in\{\pm1\}.
\end{align}

\begin{lemma}\label{SM_lemma_StateLearning_0}
Let $a_1,\ldots,a_N,b_1,\ldots,b_N\in\mathbb{R}$. Suppose that for every $m$, $1+b_m+a_m \ge 0$ and $1+b_m-a_m \ge 0$.
Then
\begin{align}
    \left[
    1-\mathds{E}_{r=\pm1}
    \prod_{m=1}^{N}(1+r a_m+b_m)
    \right]_+
    \le
    \sum_{m=1}^{N}\left(a_m^2+2|b_m|\right),
\end{align}
where $[x]_+ := \max\{x,0\}$.
\end{lemma}
Applying Lemma~\ref{SM_lemma_StateLearning_0} with $a_m=\epsilon_1G_z^{\vb{o}_{<m}}(o_m)$ and $b_m=\epsilon_1^2H_z^{\vb{o}_{<m}}(o_m)$ gives
\begin{align}
    \TVD(Q,R_z)
    &=
    \sum_{\vb{o}}
    \left[
        Q(\vb{o})-R_z(\vb{o})
    \right]_+ \leq
    \epsilon_1^2
    \sum_{\vb{o}}Q(\vb{o})
    \sum_{m=1}^{N}
    \left[
        \left(G_z^{\vb{o}_{<m}}(o_m)\right)^2
        +
        2\left|H_z^{\vb{o}_{<m}}(o_m)\right|
    \right].
\end{align}
Because $Q$ is independent of $z$, averaging over the nonzero displacement indices gives
\begin{align}
    \mathds{E}_{z\neq(0,0)}\TVD(Q,R_z)
    &\leq
    \epsilon_1^2
    \sum_{\vb{o}}Q(\vb{o})
    \sum_{m=1}^{N}
    \left[
        \mathds{E}_{z\neq(0,0)}\left(G_z^{\vb{o}_{<m}}(o_m)\right)^2
        +
        2\mathds{E}_{z\neq(0,0)}\left|H_z^{\vb{o}_{<m}}(o_m)\right|
    \right].
\end{align} 
It remains to bound the two averaged terms in the last expression.

\begin{lemma}\label{SM_lemma_StateLearning_1}
For every round $m$, every sequence of previous measurement outcomes $\vb{o}_{<m}$, and every outcome $o_m$, with $G_z^{\vb{o}_{<m}}(o_m)=0$ when $\Tr M_{o_m}^{\vb{o}_{<m}}=0$, one has
\begin{align}
    \mathds{E}_{z\neq(0,0)}\!\left[\left(G_z^{\vb{o}_{<m}}(o_m)\right)^2\right]
    \leq
    \frac{6}{d}.
\end{align}
\end{lemma}

\begin{lemma}\label{SM_lemma_StateLearning_2}
Suppose that $\Tr M_{o_m}^{\vb{o}_{<m}}>0$ and that
\begin{align}
    \frac{M_{o_m}^{\vb{o}_{<m}}}{\Tr M_{o_m}^{\vb{o}_{<m}}}
    \leq
    2^\eta I_{S_m}\otimes\sigma_{o_m}^{\vb{o}_{<m}}
\end{align}
for some density operator $\sigma_{o_m}^{\vb{o}_{<m}}$ on $A_m$.
Then
\begin{align}
    \mathds{E}_{z\neq(0,0)}\!\left[\left|H_z^{\vb{o}_{<m}}(o_m)\right|\right]
    \leq
    4\sqrt{\frac{2^\eta}{d}}.
\end{align}
\end{lemma}

Applying Lemmas~\ref{SM_lemma_StateLearning_1} and~\ref{SM_lemma_StateLearning_2} and using $\sum_{\vb{o}}Q(\vb{o})=1$ gives
\begin{align}
    \mathds{E}_{z\neq(0,0)}\TVD(Q,R_z)
    &\leq
    \epsilon_1^2N
    \left[
        \frac{6}{d}
        +
        8\sqrt{\frac{2^\eta}{d}}
    \right]
    =
    18N\epsilon^2
    \left[
        \frac{6}{d}
        +
        8\sqrt{\frac{2^\eta}{d}}
    \right].
\end{align}
For $d\geq2$ and $\eta\geq0$, we have $6/d\leq6\sqrt{2^\eta/d}$, and hence
\begin{align}
    \mathds{E}_{z\neq(0,0)}\TVD(Q,R_z)
    \leq
    252N\epsilon^2\sqrt{\frac{2^\eta}{d}}.
\end{align}
Combining this estimate with Eq.~\eqref{Eq_SM_StateLearning_TVD_Lower} yields
\begin{align}
    N
    =
    \Omega\!\left(
        \epsilon^{-2}\sqrt{d\,2^{-\eta}}
    \right).
\end{align}
This proves Corollary~2.
Setting $\eta=0$ yields Theorem~3 because every normalized POVM effect satisfying the reduction criterion obeys the corresponding $\eta=0$ operator inequality.

\subsection{Auxiliary lemmas}

\subsubsection{Proof of Lemma~\ref{SM_lemma_StateLearning_0}}

\begin{proof}
Define
\begin{align}
A_{\pm}
&:=
\prod_{m=1}^{N}
\left(
1\pm a_m+b_m
\right),
\qquad
s_m
:=
\sqrt{
(1+b_m)^2-a_m^2
}.
\end{align}
The assumptions imply $A_{\pm}\geq0$ and $s_m\geq0$.
The arithmetic--geometric mean inequality therefore gives
\begin{align}
\mathds{E}_{r=\pm1}
\prod_{m=1}^{N}
\left(
1+r a_m+b_m
\right)
=
\frac{A_++A_-}{2}
\geq
\sqrt{A_+A_-}
=
\prod_{m=1}^{N}s_m.
\end{align}
Setting $u_m:=\min\{s_m,1\}$, we obtain
\begin{align}
\left[
1-
\mathds{E}_{r=\pm1}
\prod_{m=1}^{N}
\left(
1+r a_m+b_m
\right)
\right]_+
&\leq
\left[
1-
\prod_{m=1}^{N}s_m
\right]_+
\leq
1-
\prod_{m=1}^{N}u_m
\leq
\sum_{m=1}^{N}
[1-s_m]_+.
\end{align}
If $s_m\geq1$, then $[1-s_m]_+=0$.
If $s_m<1$, then
\begin{align}
[1-s_m]_+
=
\frac{1-s_m^2}{1+s_m}
\leq
1-s_m^2
=
a_m^2-2b_m-b_m^2
\leq
a_m^2+2|b_m|.
\end{align}
Summing over $m$ proves the claim.
\end{proof}

\subsubsection{Proof of Lemma~\ref{SM_lemma_StateLearning_1}}

\begin{proof}
Fix $m$, $\vb{o}_{<m}$, and $o_m$, and write
$M:=M_{o_m}^{\vb{o}_{<m}}$.
The claim is immediate if $\Tr M=0$.
Otherwise, define
\begin{align}
\tau
&:=
\frac{M}{\Tr M},
\qquad
\tau_S
:=
\Tr_A\tau,
\qquad
\tau_A
:=
\Tr_S\tau.
\end{align}
Then
\begin{align}
G_z^{\vb{o}_{<m}}(o_m)
&=
\Tr(E_z\tau_S)
+
\Tr(E_{\overline z}\tau_A),
\qquad
\left(
G_z^{\vb{o}_{<m}}(o_m)
\right)^2
\leq
2\left[\Tr(E_z\tau_S)\right]^2
+
2\left[\Tr(E_{\overline z}\tau_A)\right]^2.
\end{align}
The orthogonality relation
$\Tr(E_zE_{z'})=d\delta_{z,z'}$ implies, for every Hermitian operator $X$,
\begin{align}
\mathds{E}_{z}
\left[
\left(
\Tr(E_zX)
\right)^2
\right]
=
\frac{1}{d}\Tr(X^2).
\end{align}
Since $z\mapsto\overline z$ is a bijection of $\mathbb Z_d^2$, it follows that
\begin{align}
\mathds{E}_{z}
\left[
\left(
G_z^{\vb{o}_{<m}}(o_m)
\right)^2
\right]
\leq
\frac{2}{d}\Tr(\tau_S^2)
+
\frac{2}{d}\Tr(\tau_A^2)
\leq
\frac{4}{d},
\qquad
\mathds{E}_{z\neq(0,0)}
\left[
\left(
G_z^{\vb{o}_{<m}}(o_m)
\right)^2
\right]
\leq
\frac{d^2}{d^2-1}
\frac{4}{d}
\leq
\frac{6}{d}.
\end{align}
This proves Lemma~\ref{SM_lemma_StateLearning_1}.
\end{proof}

\subsubsection{Proof of Lemma~\ref{SM_lemma_StateLearning_2}}

\begin{proof}
For brevity, write
$S:=S_m$, $A:=A_m$,
$M:=M_{o_m}^{\vb{o}_{<m}}$,
$\tau:=M/\Tr M$, and
$\sigma_A:=\sigma_{o_m}^{\vb{o}_{<m}}$.
The assumption gives
$\tau\leq2^\eta I_S\otimes\sigma_A$.

For $z=(q,p)$, define
\begin{align}
A_z
&:=
D_{q,p}\otimes D_{-q,p},
\qquad
B_z
:=
D_{q,p}\otimes D_{q,-p},
\qquad
a_z
:=
\Tr(A_z\tau),
\qquad
b_z
:=
\Tr(B_z\tau).
\end{align}
Using the definition of $E_z$, we have
\begin{align}
E_z\otimes E_{\overline z}
&=
\frac{i}{2}
\left(
A_z-A_z^{\dagger}
\right)
+
\frac{1}{2}
\left(
B_z+B_z^{\dagger}
\right),
\qquad
|H_z|
\leq
|a_z|+|b_z|.
\label{Eq_SM_H_bound_ab}
\end{align}

Define
\begin{align}
\alpha_{a,b}
&:=
\bra{\Phi_{a,b}}\tau\ket{\Phi_{a,b}},
\qquad
\beta_{a,b}
:=
\bra{\Psi_{a,b}}\tau\ket{\Psi_{a,b}}.
\end{align}
The diagonalizations from the preliminaries give
\begin{align}
a_{q,p}
&=
\sum_{a,b\in\mathbb Z_d}
\omega^{ap-bq}\alpha_{a,b},
\qquad
b_{q,p}
=
\sum_{a,b\in\mathbb Z_d}
\omega^{ap-bq}\beta_{a,b}.
\end{align}
Since both Bell families are orthonormal bases and every Bell state is maximally entangled,
\begin{align}
\sum_{a,b}\alpha_{a,b}
=
\sum_{a,b}\beta_{a,b}
=
1,
\qquad
0
\leq
\alpha_{a,b},\beta_{a,b}
\leq
\frac{2^\eta}{d},
\qquad
\sum_{a,b}\alpha_{a,b}^2
\leq
\frac{2^\eta}{d},
\qquad
\sum_{a,b}\beta_{a,b}^2
\leq
\frac{2^\eta}{d}.
\end{align}
The Parseval--Plancherel identity~\cite{Terras1999} therefore yields
\begin{align}
\sum_{q,p}|a_{q,p}|^2
&\leq
2^\eta d,
\qquad
\sum_{q,p}|b_{q,p}|^2
\leq
2^\eta d.
\end{align}
Combining Eq.~\eqref{Eq_SM_H_bound_ab} with the Cauchy--Schwarz inequality gives
\begin{align}
\mathds{E}_{z\neq(0,0)}|H_z|
\leq
\frac{1}{d^2-1}
\sum_{z\neq(0,0)}
\left(
|a_z|+|b_z|
\right)
\leq
\frac{1}{\sqrt{d^2-1}}
\left[
\left(
\sum_z|a_z|^2
\right)^{1/2}
+
\left(
\sum_z|b_z|^2
\right)^{1/2}
\right]
\leq
4\sqrt{\frac{2^\eta}{d}}.
\end{align}
\end{proof}

\bibliography{reference}